\documentclass[11pt,letterpaper]{article}
\usepackage[utf8]{inputenc}
\usepackage[margin=1.0in]{geometry}
\usepackage{palatino}
\usepackage{macros}
\usepackage{csquotes,fullpage}

\newcommand{\ER}{Erdős–Rényi}

\newcommand{\nul}{\mathcal{N}}
\newcommand{\pla}{\mathcal{P}}
\newcommand{\adv}{\textup{R}}
\newcommand{\algadv}{\textup{Adv}}
\newcommand{\F}{\mathcal{F}}
\newcommand{\alg}{\textup{alg}}
\newcommand{\Span}{\textup{Span}}
\newcommand{\sym}{\textup{sym}}
\newcommand{\Var}{\textup{Var}}

\newcommand{\rad}{\textup{Rad}}
\newcommand{\clique}{\textup{clique}}

\title{On optimal distinguishers for Planted Clique}
\author{Ansh Nagda\thanks{UC Berkeley, \url{anshnagda@gmail.com}, Research supported by NSF CCF-2342192.
}\and Prasad Raghavendra\thanks{UC Berkeley, \url{raghavendra@berkeley.edu}, Research supported by NSF CCF-2342192.
}}
\date{}
\begin{document}
\maketitle 
\begin{abstract}

In a {\it distinguishing problem}, the input is a sample drawn from one of two distributions and the algorithm is tasked with identifying the source distribution.
The performance of a distinguishing algorithm is measured by its advantage, i.e., its incremental probability of success over a random guess.

A classic example of a distinguishing problem is the Planted Clique problem, where the input is a graph sampled from either $G(n,1/2)$ -- the standard Erdős–Rényi model, or $G(n,1/2,k)$ -- the Erdős–Rényi model with a clique planted on a random subset of $k$ vertices. 
The \emph{Planted Clique Hypothesis} asserts that efficient algorithms cannot achieve advantage better than some absolute constant, say $1/4$, whenever $k=n^{1/2-\Omega(1)}$.

In this work, we aim to precisely understand the optimal distinguishing advantage achievable by efficient algorithms on Planted Clique.
We show the following results under the Planted Clique hypothesis:
\begin{itemize}
    \item \textbf{Optimality of low-degree polynomials:} No efficient algorithm can beat the advantage the optimal low-degree polynomial. Concretely, this means that the advantage of any efficient algorithm is at most $(1+o(1))\cdot k^2/(\sqrt{\pi}n)$, which is optimal in light of a simple edge-counting algorithm achieving this bound.

    \item \textbf{Harder planted distributions:} There is an efficiently sampleable distribution $\pla^*$ supported on graphs containing $k$-cliques such that no efficient algorithm can distinguish $\pla^*$ from $G(n,1/2)$ with advantage $n^{-d}$ for an arbitrarily large constant $d$. In other words, there exist alternate planted distributions that are much harder than $G(n,1/2,k)$.
\end{itemize}
Along the way, we prove a constructive hard-core lemma for a broad class of distributions with respect to low-degree polynomials. This result is applicable much more widely beyond Planted Clique and might be of independent interest.

\end{abstract}
\thispagestyle{empty}

\newpage

\clearpage
\setcounter{page}{1}

\section{Introduction}

In a {\it distinguishing problem}, one receives a sample drawn from one of two distributions, which we will refer to as the {\it ``planted"} and {\it ``null''} distributions. The task of the algorithm is to determine the source distribution.

The archetypal example of a {\it distinguishing problem} is the Planted Clique problem. Here the null distribution $\nul=G(n,1/2)$ is the \ER\ model over $n$-vertex graphs, where each edge is included independently with probability $\frac{1}{2}$. The planted distribution $\pla=G(n,1/2,k)$ is parametrized by a number $0\leq k\leq n$, and is sampled as follows.\footnote{This planted distribution has been referred to as the {\it binomial} $k$ model \cite{hajek2015computational,bresler2023detection}. See \cite{hirahara2024planted} for a comparison with other models.}
\begin{itemize}
    \item Sample a random subset $S\subseteq [n]$ of vertices by including each vertex independently with probability $\frac{k}{n}$.
    \item Generate an \ER\  random graph $G\sim G(n,1/2)$.
    \item Plant a clique amongst the vertices in $S$. That is, set $G_{ij}=1$ for all $i,j\in S$.
\end{itemize}

We will focus on the computationally hard regime where $k = n^{1/2-\alpha}$ for some $0<\alpha<1/2$. In this regime, the Planted Clique problem is information theoretically easy, but has long been hypothesized to be intractable for polynomial-time algorithms \cite{jerrum1992large,kuvcera1995expected, alon1998finding,krivelevich2002approximating}.  

The intractability of distinguishing problems is quantified in terms of the \emph{distinguishing advantage}. Let $A$ be a \emph{test}, that is, a randomized mapping $A$ from the set of $n$-vertex graphs to $\{\pm 1\}$. We will define the \emph{advantage} achieved by $A$ as
\[\algadv^{(\pla,\nul)} (A):=|\E_{\pla}[A] - \E_{\nul}[A]| = 2\cdot |\Pr_{x\sim\pla}[A(x)=1] - \Pr_{x\sim\nul}[A(x)=1]|\in [0,2]\footnote{Our choice of normalization is unconventional -- advantage is usually defined as \emph{half} of this quantity. Nevertheless, this will be more notationally convenient.}.\]

Above, the expectation is taken over both $x$ and the internal randomness of $A$. In this terminology, the Planted Clique Hypothesis can be stated as follows.
\begin{conjecture} (Planted Clique Hypothesis) \label{conj:planted}
Let $\alpha>0$ and set $k=n^{1/2-\alpha}$. Let $\nul=G(n,1/2)$ and $\pla = G(n,1/2,k)$. For every randomized polynomial time test $A$ and all sufficiently large $n$,
\[ \algadv^{(\pla,\nul)}(A) \leq 1/4.\]
\end{conjecture}
In light of the seeming intractability of the Planted Clique problem, it is natural to ask the following question.
\begin{quote}
    \centering {\it What is the optimal (i.e., smallest) advantage achievable by any efficient algorithm?}
\end{quote}
In an elegant sequence of works, Hirahara and Shimizu \cite{hirahara2023hardness, hirahara2024planted} demonstrated that random self-reductions can be used to amplify the hardness of the Planted Clique problem.  In particular, these works showed that  \cref{conj:planted} implies the following stronger bound on the advantage.

\begin{theorem}\label{thm:hirahara} (Theorem 1.3, (7) in \cite{hirahara2024planted})
    Assume the Planted Clique Hypothesis. Let $\alpha,\gamma>0$ and set $k=n^{1/2-\alpha}$. Let $\nul=G(n,1/2)$ and $\pla = G(n,1/2,k)$. For every randomized polynomial time test $A$, 
\[ \algadv^{(\pla,\nul)}(A) = n^{\gamma}\cdot O\left(\frac{k^2}{n}\right).\]
\end{theorem}

Notice that on average, graphs sampled from the planted distribution $\pla$ have $\binom{k}{2}/2$ more edges than an \ER\  graph sampled from $\nul$, owing to the clique planted among $k$ vertices.  This slight mismatch in the expected number of edges can be used to efficiently distinguish the planted and null distributions with an advantage of $\Omega\left(\frac{k^2}{n}\right)$.  
Thus the above bound on the distinguishing advantage is correct up to an arbitrary polynomial factor in $n$.

This leaves the question of whether we can characterize the optimal advantage for Planted Clique?
Subsequently, one may ask which class of algorithms are optimal distinguishers for the problem? Conversely, can we construct alternate planted distributions where the optimal advantage is significantly smaller than $\frac{k^2}{n}$, say $o(n^{-100})$?
This work make progress towards answering all these questions. In order to explain our results, we will first give an introduction to the \emph{low-degree method}.

\paragraph{Low-degree polynomials.}

Starting with the works of Hopkins and Steurer \cite{Hopkins18thesis,hopkins2017efficient},  the low-degree method has emerged as an effective heuristic to determine whether a distinguishing problem is computationally tractable or not.  Roughly speaking, the low-degree method is based on the hypothesis that if all degree $D$ polynomials fail on a distinguishing task, then all algorithms running in time $2^D$ fail too.

To elucidate further, we will first set up some terminology on real-valued test functions.
Fix a distinguishing problem between a planted distribution $\pla$ and a null distribution $\nul$, both over a set $\Omega$.
For a real-valued test function $f:\Omega\to \R$, let us define $\adv^{(\pla,\nul)}(f)$ as the ratio
\[\adv^{(\pla,\nul)}(f):=\frac{\E_{\pla}[f] - \E_{\nul}[f]}{\|f\|_{2,\nul}} = \frac{\E_{\pla}[f] - \E_{\nul}[f]}{\sqrt{\E_\nul[f^2]}}\in \R.\]
Informally, one can think of $\adv^{(\pla,\nul)}(f)$ measuring the ``performance'' of $f$ at distinguishing between $\pla$ and $\nul$ -- note that if the output of $f$ is in $\{\pm 1\}$, then one recovers the advantage
\[|\adv^{(\pla,\nul)}(f)| = \algadv^{(\pla,\nul)}(f).\]
$\adv^{(\pla,\nul)}(f)$ is not as readily interpreted if the output of $f$ is \emph{not} restricted to $\{\pm 1\}$, but we will nevertheless treat it as a sort of distinguishing advantage.
For a family of test functions $\F\subseteq L^2(\nul)$, let $\adv^{(\pla,\nul)}[\F]$ denote the optimal advantage achieved by a function in $\F$, i.e.,
\[\adv^{(\pla,\nul)}[\F] := \sup_{f\in \F}\adv^{(\pla,\nul)}(f).\]
%
%
For simplicity, let us restrict our attention to the case where $\Omega=\{\pm 1\}^N$ is the $N$-dimensional hypercube for some $N \in \N$. In our running example of Planted Clique, one should think of $N = \binom{n}{2}$; a graph $G$ will be represented by its flattened $\{\pm 1\}$-valued adjacency matrix in $\{\pm 1\}^{\binom{n}{2}}$.
%
Let $\F_{\leq d}$ denote the subspace of polynomials of degree at most $d$ over $\Omega$.
The quantity $\adv^{(\pla,\nul)}[\F_{\leq d}]$ denotes the ``advantage'' achieved by the best polynomial of degree at most $d$, and we will refer to it as the \emph{low-degree advantage}.
A particularly nice feature of the low-degree advantage is that it can often be explicitly computed for many distinguishing problems using analytic techniques.
%
%
The central thesis behind the low-degree method is the following conjecture:
\begin{conjecture}[Low-degree conjecture, informal]
    For many  ``natural'' distinguishing tasks between $\pla$ and $\nul$, $\adv^{(\pla,\nul)}[\F_{\leq O(\log N)}] = o(1)$ implies that for every polynomial time test $A$, $\algadv^{(\pla,\nul)}(A) = o(1)$.
\end{conjecture}
We refer the reader to \cite{Hopkins18thesis,KWB19} for a more precise statement.  
\paragraph{Computational hardness based on low-degree advantage.}
In light of the low-degree conjecture, it may be natural to hypothesize that low-degree polynomials yield the optimal advantage among all efficient algorithms.
For Planted Clique, our results support this hypothesis by proving that under \cref{conj:planted}, no efficient algorithm can achieve advantage better than the low-degree advantage. Specifically, we show the following.
\begin{theorem}\label{thm:main-1}
 Assume the Planted Clique Hypothesis (\cref{conj:planted}). Let $\alpha>0$ and set $k=n^{1/2-\alpha}$. Let $d\in \N$ be a sufficiently large constant. Let $\nul=G(n,1/2)$ and $\pla = G(n,1/2,k)$. For every randomized polynomial time test $A$,
\[ \algadv^{(\pla,\nul)}(A) \leq (1+ o_{n}(1)) \cdot \adv^{(\pla,\nul)} [ \F_{ \leq d}].\]
\end{theorem}
As stated before, the low-degree advantage $\adv^{(\pla,\nul)} [ \F_{ \leq d}] $ can be computed explicitly, giving a concrete upper bound of about $k^2/(\sqrt{\pi}n)$ on the advantage. In fact, our proof can be refined slightly to yield the following nearly optimal bound.

\begin{corollary}\label{cor:ldlr-value}
Assume the Planted Clique Hypothesis (\cref{conj:planted}). Let $\alpha>0$ and set $k=n^{1/2-\alpha}$. Let $\nul=G(n,1/2)$ and $\pla = G(n,1/2,k)$. For every randomized polynomial time test $A$,
\[ \algadv^{(\pla,\nul)}(A) \leq (1+o_n(1))\cdot \frac{k^2}{\sqrt{\pi} n}.\]
\end{corollary}

\begin{remark}[Optimality]
Conversely, one can show that a simple algorithm based on a low-degree polynomial achieves this bound -- testing whether the input graph has more than $\binom{n}{2}/2$ edges achieves advantage $\approx \frac{k^2}{\sqrt{\pi}n}$, so the above hardness result is tight up to $1+o_n(1)$ factors.
\end{remark}

The above results fall out as a special case of a much more general result that can be applied beyond the Planted Clique problem.
 In particular, we prove the above result for any distinguishing problems where the underlying planted measure satisfies a few natural properties such as hypercontractivity of low-degree polynomials; we expect that these properties hold in many settings where the planted distribution is obtained by adding a sparse signal onto the null distribution.  We defer the full statement of the general theorem to \cref{thm:general-1}.

Another immediate consequence of our general result (\cref{thm:general-1}) is that the optimality of low-degree polynomials holds even for perturbations of the planted distribution. 
Suppose $\pla'$ is a small perturbation of the planted distribution $\pla=G(n,1/2,k)$ in that the relative density $\frac{d\pla'}{d\pla}$ is bounded pointwise by $\poly(n)$. Let $\pla^*$ be a ``noisy'' version of $\pla'$ in the sense that it is obtained by adding a small amount of additional randomness to $\pla'$.

\begin{theorem}[Informal version of \cref{thm:perturbation-full}] \label{thm:perturbation}
Assume the Planted Clique hypothesis (\cref{conj:planted}). Let $\pla'$ and $\pla^*$ be as above, and
let $d$ be a sufficiently large constant. For any efficiently computable test $A$,
\[ \algadv^{(\pla^*, \nul)}(A) \leq \adv^{(\pla^*,\nul)}[\F_{\leq d}] + n^{-\Omega(\sqrt{d})} \cdot \left\| \frac{d \pla'}{d \pla} \right\|_{\infty} \]
\end{theorem}

For the sake of readability, here we do not elaborate on the details of the aforementioned noise-addition process, other than to remark that it is done in a way that approximately preserves clique size. That is, $\pla^*$ is still essentially supported on graphs containing large cliques.

\paragraph{A Hard-Core Distribution.}

As noted earlier, the natural planted distribution $\pla$ can be distinguished from the null distribution with an advantage of $\Omega(\frac{k^2}{n}) \gg \frac{1}{n}$ simply by counting the number of edges. 
Motivated by recent applications of Planted Clique-like problems to cryptography \cite{abram2023cryptography,bogdanov2024low}, it is natural to ask whether there exist other planted distributions over graphs with $k$-cliques that make the distinguishing problem much harder. For instance, does there exist a planted distribution $\pla^*$ such that \[\algadv^{(\pla^*,\nul)}(A)=o(1/n^{100})\] for efficient tests $A$?

This is analogous to Impagliazzo's hard-core lemma \cite{I95} in complexity theory, wherein starting with the assumption that a function $f$ is hard to compute with probability $0.99$ for circuits of size $s$, one can construct a distribution $\mathcal{D}$ over inputs under which the function cannot be computed with probability $\frac{1}{2}+\delta$ for slightly smaller circuits.

Our general theorem on perturbations of the planted distribution (see \cref{thm:perturbation}) paves the way towards constructing a hard distribution $\pla^*$ by giving a way to turn low-degree hardness into computational hardness. 
We begin by constructing a small perturbation $\pla'$ of $\pla$ so that the low-degree advantage $\adv^{(\pla',\nul)}[\F_{\leq d}]$ is at most a desired bound $\delta$.
Specifically, we will show the following.
\begin{theorem}\label{thm:intro-moment-match}
     For any constants $\alpha>0$ and $d\in \N$, and any $\delta=\delta(n) > 0$, there is a $\poly(n^d, 1/\delta)$ time algorithm to sample from a distribution $\pla'$ satisfying the following:
    \begin{enumerate}
        \item The low-degree advantage for the $(\pla',\nul)$ distinguishing problem is at most $\delta$, i.e., $\adv^{(\pla',\nul)}[\F_{\leq d}] \leq \delta$.
        \item $\pla'$ is a small perturbation of $\pla$, i.e., $\left\|\frac{d\pla'}{d\pla}\right\|_\infty\leq 1+o(1)$.
    \end{enumerate}
\end{theorem}

The above result can be viewed as a constructive hard-core lemma against low-degree polynomials.
Let $\pla^*$ be a noisy version of $\pla'$ as in the setup for \cref{thm:perturbation}. Appealing to \cref{thm:perturbation}, we conclude the following computational hardness result, affirmatively answering the above question.
\begin{theorem}\label{thm:main-clique}
    Assume the Planted Clique Hypothesis (\cref{conj:planted}). For any $0<\alpha<1/2$ and $D\in \N$, there is a $\poly(n)$ time algorithm to sample from a distribution $\pla^*$ supported on graphs containing a clique of size at least $k=n^{1/2-\alpha}$ such that for any polynomial time test $A$,
        \[\algadv^{(\pla^*,\nul)}(A) = o(n^{-D}).\]
\end{theorem}

Finally, we note that our proof of \cref{thm:intro-moment-match} generalizes to a much broader class of distinguishing problems beyond Planted Clique, and may be of independent interest. Let us state an informal version of this generalization.

\begin{theorem}[Informal version of \cref{thm:general-2}]
Let $\pla$ and $\nul$ be two distributions over $\Omega$, and let $V$ be an arbitrary subspace of $L^2(\nul)$ such that $V$ is computationally ``tractable'', and $V$ satisfies a hypercontractive inequality. If $\adv^{(\pla,\nul)}[V] = o(1)$, then for any $d\in \N$ one can efficiently find a small perturbation $\pla'$ of $\pla$ in the sense that $\|\frac{d\pla'}{d\pla}\|_\infty = 1+o(1)$, such that $\adv^{(\pla,\nul)}[V]=1/n^d$.
\end{theorem}

We formally state and prove this result in \cref{sec:hardcore}.

\subsection{Related Work}
\paragraph{Planted Clique.}
The planted clique problem was originally proposed by Jerrum~\cite{jerrum1992large} and Kučera~\cite{kuvcera1995expected} as a search version wherein the task was to identify a $k$-clique hidden within an $n$-vertex Erd\H{o}s--R\'enyi random graph.  The decision version of the problem was suggested by Saks \cite{alon1998finding,krivelevich2002approximating} has since resisted all algorithmic attacks.  For $k = o(\sqrt{n})$, there is a growing body of lower bounds in restricted models of computation which suggest that the problem is intractable.  We refer the reader to \cite{hirahara2024planted} for a survey on the Planted Clique problem.

Recognized as one of the most prominent problems in average-case complexity, its presumed computational difficulty has led to wide-ranging applications in areas including average-case complexity~\cite{elrazik2022pseudorandom}, cryptography~\cite{juels2000hiding,applebaum2010public,abram2023cryptography}, hardness of approximation~\cite{manurangsi2020strongish}, game theory~\cite{hazan2011hard} and property testing~\cite{alon2007testing}.  More recently, Planted Clique and its variants have been used to develop a web of reductions amongst problems in high-dimensional statistics \cite{berthet2013complexity, DBLP:conf/colt/BrennanBH18, BB20, BR13, MW15, Che15, HWX15, WBS16, GMZ17, CLR17, WBP16, ZX18, CW18, , BB19, }.

\paragraph{Low Degree Method.}
The idea of using low-degree polynomials as benchmark algorithms for hypothesis testing was first explored in the work of Barak et al. \cite{BHKKMP19} on sum-of-squares lower bounds for Planted Clique. This was then further developed by Hopkins and Steurer \cite{hopkins2017efficient}, and subsequently by others (see, e.g. \cite{HopkinsKPRSS17, Hopkins18thesis}, and also \cite{KWB19} for a survey). This technique has since become known as the "low-degree hardness" paradigm, and has been applied to a variety of statistical problems \cite{BHKKMP19, hopkins2017efficient, HopkinsKPRSS17, Hopkins18thesis,BKW19, KWB19, DKWB23, BB19, MRX19, CHKRT19}.

Perhaps the work that is most related to ours is the the work of Moitra and Wein \cite{moitra2023preciseerrorratescomputationally} that studies the precise error rates for the Spiked Wigner model.
Towards exactly characterizing the asymptotic error achievable by efficient algorithms, they propose the following strengthening of the low-degree conjecture, and investigate its consequences for the \emph{spiked Wigner model}.
\begin{conjecture}[Moitra-Wein strengthening of the Low-degree conjecture, informal]
    Let $\F_\alg^\R$ be the class of real-valued functions $f\in L^2(\nul)$ computable in polynomial time. For many ``natural'' distinguishing tasks between $\pla$ and $\nul$, \[\adv^{(\pla,\nul)}[\F_\alg^\R]\leq (1+o(1))\cdot \adv^{(\pla,\nul)}[\F_{O(\log n)}].\]
\end{conjecture}

More recently, Ding, Hua, Slot, and Steurer \cite{ding2025low} also leveraged a similar strengthening of the Low-degree conjecture for the stochastic block model. Unfortunately due to technical issues, the above conjecture turns out to be false as stated even for the sparse Spiked Wigner and Planted Clique models, and needs to be amended. One possible amendment states that the above holds for $f$ belonging to a certain ``nice'' set of efficiently computable real-valued functions that seem to circumvent these technical issues.
We believe that \cref{thm:main-1,thm:perturbation} nicely complement this hypothesis that under suitable restrictions, the behaviour of low-degree polynomials characterizes the advantage of optimal distinguishers.


\paragraph{Hard-Core Lemma and Hardness amplification.}
Fix a class $\mathcal{C}$ of circuits
and  a boolean
function $f : \{0, 1\}^n \to \{0, 1\}$ which is somewhat hard to
compute for the class $\mathcal{C}$, in that no circuit from $\mathcal{C}$ succeeds with probability $> 1-\delta$.  The (nonuniform) hard-core lemma asserts that there exists
a large subset of inputs on which the function $f$ is essentially unpredictable for a slightly smaller class of circuits.
Impagliazzo \cite{I95}  gave the first proof of such a theorem
and used it to derive an alternative proof of Yao’s XOR
lemma \cite{yao1982theory}. 
%

%
Constructive versions of the hard-core lemma can be obtained by a boosting
algorithm \cite{klivans1999boosting,barak2009uniform}.

A technical point is that in general, one cannot prove exactly analogous statements against the class of uniform polynomial-time algorithms. Holenstein \cite{holenstein2005key} proved a useful statement that can be used to port many of the key consequences of the nonuniform hard-core lemma to the uniform setting. In particular, assuming $f$ is somewhat hard to compute for polynomial-time algorithms, for any polynomial-time oracle algorithm $A^{(\cdot)}$ that can make input-independent membership queries to a set $S\subseteq \{0,1\}^n$, there is a hard-core set $S$ such that $A^{(S)}$ cannot compute $f$ on a nontrivial fraction of inputs from $S$. This statement does not guarantee the existence of a \emph{universal} hard-core set that works against all polynomial-time algorithms $A$, and this seems to be a generic limitation that cannot be worked around without nonuniform hardness assumptions.

In \cref{thm:main-clique}, we prove a constructive, uniform hard-core lemma for the Planted Clique Problem, avoiding the aforementioned generic limitation. We begin by first obtaining an assumption-free hard-core planted distribution against the class of low-degree polynomials. Then, with \cref{thm:perturbation}, we show how to lift this to a computational hardness statement by proving that a slightly noisy version of this distribution is hard for the class of efficient algorithms under the mild assumption \cref{conj:planted}.

\subsection{Paper Organization}
\cref{sec:prelims} contains preliminaries. \cref{sec:overview} contains an overview of our proof ideas. In \cref{sec:reduction} we state and prove a generalized version of reduction. In \cref{sec:hardcore} we state and prove a generalized version of our technique to find hard-core distributions . In \cref{sec:application} we provide describe how to instantiate our general theorems to recover our main results on Planted Clique (\cref{thm:main-1,thm:main-clique}).

\section{Preliminaries}\label{sec:prelims}
\paragraph{Probability and Analysis.}
For a finite set $\Omega$, $L^2(\Omega)$ denotes the vector space of all real-valued functions $f:\Omega\to \R$. Let $\nul$ be a distribution over $\Omega$. The $L^p$ norms and inner product induced by $\nul$ are defined as follows:
\[\|f\|_{p,\nul} = \E_{\nul}[f^p]^{1/p},\qquad \langle f, g\rangle_{\nul} = \E_{\nul}[f\cdot g],\]
where $f\cdot g:\Omega\to \R$ is the pointwise product $f\cdot g(x) = f(x)\cdot g(x)$. Since we will be simultaneously working with multiple distributions, we will always write the distribution's label in the subscript to avoid ambiguity. We say that $f$ and $g$ are \emph{orthogonal with respect to $\nul$} if $\langle f,g\rangle_\nul=0$. As syntactic sugar, often we will write $L^2(\nul)$ to denote $L^2(\Omega)$ with the suggestion that inputs to functions in $L^2(\nul)$ are typically drawn from $\nul$.

For any $1\leq p\leq q$, Jensen's inequality applied to the convex function $x\to x^{q/p}$ implies that $\|f\|_{p,\nul}\leq \|f\|_{q,\nul}$. The following is a consequence of Hölder's inequality.

\begin{lemma}[Log-convexity of $L^p$ norms]\label{lemma:log-convexity}
    For $f\in L^2(\nul)$, 
    \[\|f\|_{1,\nul}\geq \frac{\|f\|_{2,\nul}^3}{\|f\|_{4,\nul}^2}\quad\text{ and }\quad\frac{\|f\|_{4,\nul}}{\|f\|_{2,\nul}}\leq \frac{\|f\|_{8,\nul}^2}{\|f\|_{4,\nul}^2}.\]
\end{lemma}

\begin{definition}[$(a,b)$-anticoncentrated]\label{def:anticoncentration}
    A real-valued random variable $x$ is $(a,b)$-anticoncentrated if $\Pr[|x|\geq b\cdot \E[x^2]^{1/2}]\geq a$.
\end{definition}

\begin{lemma}[Paley-Zygmund inequality]\label{lem:paley-zygmund}
    A real-valued random variable $x$ is $\left((1-\sqrt{\theta})^2\cdot \frac{\E[x^2]^2}{\E[x^4]},\theta\right)$-anticoncentrated for any $\theta\in [0,1]$.
\end{lemma}

\begin{definition}[Pointwise product]
    For two subspaces $V,V'\subset L^2(\Omega)$, the pointwise product $V\odot V'$ is the space consisting of $fg$ for $f\in V$ and $g\in V'$. Define $V^{\odot k} = \underbrace{V\odot \ldots \odot V}_{k\text{ times}}$.
\end{definition}

The Bernoulli distribution $\ber(p)$ is the distribution over $\{0,1\}$ that outputs $1$ with probability $p$ and $0$ with probability $1-p$. The Rademacher distribution $\rad$ is the uniform distribution over $\{\pm 1\}$.

\paragraph{Distinguishing Problems.}

For two distributions $\pla$ and $\nul$ over $\Omega$, let $\frac{d\pla}{d\nul}\in L^2(\nul)$ be the relative density $\frac{d\pla}{d\nul}(x) = \frac{\Pr_\pla[x]}{\Pr_\nul[x]}$. For a function $f\in L^2(\nul)$ define
\[\adv^{(\pla,\nul)}(f) = \frac{\E_\pla[f]-\E_\nul[f]}{\|f\|_{2,\nul}}.\]
Let $\F\subseteq L^2(\nul)$. Define
\[\adv^{(\pla,\nul)}[\F]=\sup_{f\in \F}\adv^{(\pla,\nul)}(f).\]

The following claim gives a way to compute $\adv^{(\pla,\nul)}[V]$ for subspaces $V$ in terms of an orthonormal eigenbasis of $V$.

\begin{claim}\label{claim:ldlr-expression}
    Let $V\subseteq L^2(\nul)$ be a subspace containing $\mathbf{1}$, and let $\mathbf{1},f_1,\ldots, f_\ell$ be a basis that is orthonormal with respect to $\langle \cdot,\cdot\rangle_\nul$. Then, $\adv^{(\pla,\nul)}[V] = \sqrt{\sum_{i\in [\ell]}\E_{\pla}[f_i]^2}$.
\end{claim}
\begin{proof}
    Let $f\in V$. Expand $f=c_0\mathbf{1} + \sum_{i\in [\ell]}c_i f_i$. We have
    \begin{align*}
      \adv^{(\pla,\nul)}(f) &= \frac{\sum_{i\in [\ell]}c_i\cdot \E_{\pla}[f_i]}{\sqrt{c_0^2+\sum_{i\in [\ell]}c_i^2}}\tag{$\E_\nul[f_i]=0$ for all $i\in [\ell]$ by orthogonality}\\
      &\leq \frac{\sum_{i\in [\ell]}c_i\cdot \E_{\pla}[f_i]}{\sqrt{\sum_{i\in [\ell]}c_i^2}}\tag{$c_0^2\geq 0$}\\
      &\leq \sqrt{\sum_{i\in [\ell]}\E_\pla[f_i]^2}.\tag{Cauchy-Schwarz}
    \end{align*}
    This proves $\adv^{(\pla,\nul)}[V] \leq \sqrt{\sum_{i\in [\ell]}\E_{\pla}[f_i]^2}$. For the other direction, notice that one could have picked $c_i = \E_\pla[f_i]$ and $c_0=0$, in which case all inequalities would be tight, proving $\adv^{(\pla,\nul)}[V] = \sqrt{\sum_{i\in [\ell]}\E_{\pla}[f_i]^2}$.
\end{proof}

\paragraph{Markov Chains and Noise Operators.}
A Markov chain $M$ over $\Omega$ is defined by a transition rule which maps an element $x\in \Omega$ to a probability distribution over $\Omega$. We will denote a sample from this distribution by $y\sim M(x)$. For a distribution $\nul$, $M(\nul)$ will denote the distribution of $y$ in the following process: (1) sample $x\sim \nul$, (2) sample $y\sim M(x)$. If $M(\nul)=\nul$, we say $\nul$ is a stationary distribution of $M$.

We associate a \emph{noise operator} $T:L^2(\nul)\to L^2(\nul)$ to any Markov chain $M$ that is defined by $Tf(x) = \E_{y\sim M(x)}[f(y)]$. Note that $T$ always has the constant function $\mathbf{1}(x) = 1$ as an eigenvector with eigenvalue $1$.

\begin{fact}
    Let $T$ be the noise operator associated with a reversible Markov chain with stationary distribution $\nul$. $T$ has a complete basis of eigenvectors $\mathbf{1}=f_1,f_2,\ldots, f_{|\Omega|}$ that are orthogonal with respect to $\nul$ with real eigenvalues $1=\lambda_1\geq \lambda_2\geq \ldots\geq \lambda_{|\Omega|}$.
\end{fact}

We will denote by $\lambda_{\geq c}(T)$ (resp. $\lambda_{<c}(T)$) the span of all eigenvectors with eigenvalue $\geq c$ (resp. $<c$).

\paragraph{Boolean functions.}
Often, $\Omega$ will be equal to the boolean hypercube $\{\pm 1\}^{[m]}$ with $\nul=\unif(\Omega)$. For $x\in \{\pm 1\}^{[m]}$ and a set $S\subseteq [m]$, write $\chi_S(x)=x_S = \prod_{i\in S}x_i$. The \emph{Fourier basis} $\{\chi_S\}_{S\subseteq [m]}$ forms an orthonormal basis of $L^2(\nul)$. Let $\F_{\leq d}=\Span\{\chi_S:|S|\leq d\}$ be the set of all functions having degree at most $d$.

\begin{lemma}[Bonami's Lemma \cite{bonami1970etude}]\label{lem:bonami}
    Every degree $\leq d$ function $f$ over the boolean hypercube satisfies $\|f\|_{4,\nul}\leq \sqrt{3}^d\cdot \|f\|_{2,
    \nul}$.
\end{lemma}

We will also consider the setting of the $p$-biased hypercube, where $\Omega =\{0,1\}^{[m]}$ and $\nul=\ber(p)^{\otimes m}$. For $x\in \{0,1\}^{[m]}$ and $S\subseteq [m]$, define the $p$-biased Fourier character $\chi^p_S(x) = \prod_{i\in S}\frac{x_i - p}{\sqrt{p(1-p)}}$. The Fourier basis $\{\chi^p_S\}_{S\subseteq [m]}$ is orthonormal with respect to $\nul$. Again, we say that $f\in L^2(\nul)$ has degree at most $d$ if $f\in \Span\{\chi^p_S:|S|\leq d\}$. Say that $f$ is $S_m$-symmetric if $f(x) = f(\pi(x))$ for any permutation $\pi$, where $\pi$ acts on $x$ by permuting its coordinates.
\begin{lemma}[Consequence of \cite{keevash2021global}]\label{lem:p-biased}
    Suppose $d\leq mp$. Every $S_m$-symmetric degree $\leq d$ function $f$ over the $m$-dimensional $p$-biased hypercube satisfies $\|f\|_{4,\nul}\leq 8^d\cdot \|f\|_{2,\nul}$.
\end{lemma}
\begin{proof}
    Assume without loss of generality that $p\leq 1/2$. Write $f=\sum_{|T|\leq d}\hat{f}_T\chi_T^p$, and let $I_S(f) = (p(1-p))^{-|S|}\sum_{T\supseteq S}\hat{f}_T^2$.  
    We will apply Lemma II.1.6 from \cite{keevash2021global}, which states that if $I_S(f)\leq b\cdot \|f\|_{2,\nul}^2$ for all $|S|\leq d$, then $\|f\|_{4,\nul}\leq 4^d\cdot b^{1/4}\cdot \|f\|_{2,\nul}$. In order to apply this, the key property about $S_m$-symmetric functions we will use is that $\hat{f}_T = \hat{f}_{T'}$ if $|T|=|T'|$.
    \[(p(1-p))^{|S|}I_S(f) = \sum_{T\supseteq S}\hat{f}_T^2=\E_{S'\sim \binom{[m]}{|S|}}\left[\sum_{T\supseteq S'}\hat{f}_T^2\right]= \binom{m}{|S|}^{-1}\sum_{|T|\leq d}\binom{|T|}{|S|}\cdot\hat{f}_T^2,\]
    where the second inequality used $S_m$-symmetry and the last equality used the fact that $T\supseteq S'$ with probability equal to $\binom{|T|}{|S|}/\binom{m}{|S|}$.

    
    Since $|T|$ and $|S|$ are always at most $d$, we can bound $\binom{|T|}{|S|}\leq 2^d$. Therefore
    \[I_S(f)\leq (p(1-p))^{-|S|}\cdot \binom{m}{|S|}^{-1}\cdot 2^d\cdot \|f\|_{2,\nul}^2\leq 4^d\cdot \left(\frac{|S|}{mp}\right)^{|S|}\cdot \|f\|_{2,\nul}^2\leq 4^d\cdot \|f\|_{2,\nul}^2,\]
    where we used the fact that $p\leq 1/2$ and $|S|\leq d \leq mp$. In conclusion, we have proved that $I_S(f)\leq 4^d\cdot \|f\|_{2,\nul}^2$, implying $\|f\|_{4,\nul}\leq 8^d\cdot \|f\|_{2,\nul}$.
\end{proof}

\paragraph{Graphs.} In this work, we will always represent a graph by a vector $G\in \{\pm 1\}^{\binom{[n]}{2}}$. The entries of $G$ should be interpreted as follows: the edge $\{i,j\}$ is present if and only if $G_{ij}=1$.

\begin{definition}[Random Graph Models]
    We will consider two distributions over $n$-vertex graphs. The first is the \ER\ model $G(n,1/2) = \unif(\{\pm 1\}^{\binom{[n]}{2}})$. The second is the binomial-$k$ Planted Clique model $G(n,1/2,k)$ that is sampled as follows: 
    \begin{itemize}
        \item Sample a graph $G\sim G(n,1/2)$.
        \item Sample a random vector $x\sim \ber(k/n)^{\otimes n}$. For all $i,j$ such that $x_i=x_j=1$, set $G_{ij}=1$.
        \item Output $G$.
    \end{itemize}
\end{definition}

\section{Technical Overview for Planted Clique}\label{sec:overview}
In this section we describe the main proof ideas behind our results for planted clique. We will discuss \cref{thm:main-1,thm:intro-moment-match} in \cref{sec:overview-reduction,subsec:moments} respectively.

\subsection{Hardness Self-Amplification}\label{sec:overview-reduction}
Let us begin by discussing \cref{thm:main-1} which states that the Planted Clique Hypothesis (\cref{conj:planted}) implies a much stronger hardness statement. Let $\beta>0$ be a constant, let $\Omega = \{\pm 1\}^{\binom{[n]}{2}}$, and define the distributions
\[\pla^*=G(n,1/2,n^{1/2-\beta}),\quad \nul=G(n,1/2).\]
Our goal is to show that $\pla^*$ and $\nul$ cannot be distinguished with advantage significantly larger than $\adv^{(\pla^*,\nul)}[\F_{\leq d}]$ for some constant $d$. For the sake of contradiction, assume that there is an efficient test $A$ mapping $\Omega$ to $\{\pm 1\}$ achieving advantage
\begin{equation}\label{eqn:a-adv-bound}
    \algadv^{(\pla^*,\nul)}(A)> (1+\Omega(1))\cdot \adv^{(\pla^*,\nul)}[\F_{\leq d}].
\end{equation}

Let us pick another constant $\alpha$ such that $0<\alpha<\beta$, and define $\pla = G(n,1/2,n^{1/2-\alpha})$. We will use $A$ to construct another efficient test $B$ such that
\[\algadv^{(\pla,\nul)}(B)\geq \Omega(1),\]
which contradicts the Planted Clique Hypothesis by \cref{thm:hirahara}. Our construction of $B$ proceeds in two steps. First, in the \emph{amplification} step, we construct an efficiently computable \emph{real-valued} function $g:\Omega\to \R$ achieving $\adv^{(\pla,\nul)}(g) = \omega(1)$. Next, we perform a \emph{rounding} step to convert $g$ into an efficiently computable $\{\pm 1\}$-valued test $B$ achieving $\algadv^{(\pla,\nul)}(B) = \Omega(1)$.

\begin{assumption}
We assume that $A$ is deterministic, so it can be interpreted as a function $A:\Omega\to \{\pm 1\}$. We also assume that $\E_\nul[A] = 0$. All functions $f$ we will construct from $A$ will also share this property, allowing us to write the simplified expression
\[\adv^{(\pla^*,\nul)}(f) = \frac{\E_{\pla^*}[f]}{\|f\|_{2,\nul}}.\]
These assumptions are completely avoidable, and solely serve the purpose of simplifying the presentation.
\end{assumption}

\paragraph{Step 1: Amplification.} In this part, our goal is, given $A$ satisfying \cref{eqn:a-adv-bound}, to construct a function $g\in L^2(\nul)$ that achieves $\adv^{(\pla,\nul)}(g)=\omega(1)$. As in any average case reduction, we will need an efficiently computable randomized map $M$ from $\Omega$ to $\Omega$ that maps $\nul$ to $\nul$ and $\pla$ to $\pla^*$. We will use the \emph{vertex-resampling Markov Chain} $M$. We define $M$ by specifying how to sample $H\sim M(G)$ for a fixed $G\in \Omega$. Set $p:=n^{\alpha - \beta}$. $H$ is sampled as follows.
\begin{itemize}
    \item Sample $z\sim \ber(p)^{\otimes n}$.
    \item For each $\{i,j\}\in \binom{[n]}{2}$, set $H_{\{i,j\}}=G_{\{i,j\}}$ if $z_i=z_j=1$, and independently sample $H_{\{i,j\}}\sim \rad$ otherwise.
\end{itemize}

It is not too difficult to see that $M(\nul)=\nul$. As for $\pla$, one can check that $M$ essentially independently removes vertices from the clique with probability $1-p$, and hence $M(\pla)=\pla^*$ is another Planted Clique distribution with slightly smaller clique size.

\begin{remark}
We remark that similar maps have appeared in prior work. In particular $M$ is the composition of the \emph{embedding} and \emph{shrinking} reductions used by Hirahara and Shimizu \cite{hirahara2023hardness,hirahara2024planted}. This map has also been used in a different context in the work of Brennan et al. \cite{brennan2020statistical}.
\end{remark}

Let $T:L^2(\nul)\to L^2(\nul)$ be the noise operator associated with $M$, that is, $Tf(G) = \E_{H\sim M(G)}[f(H)]$. The main observation we will need is that the Fourier basis gives an explicit orthonormal basis of eigenvectors of $T$ with respect to the inner product $\langle f,g\rangle_\nul = \E_\nul[f\cdot g]$.

\begin{fact}\label{fact:eigen}
    For all $S\subseteq \binom{[n]}{2}$, $\chi_S(G) = G_S = \prod_{\{i,j\}\in S}G_{\{i,j\}}$ is an eigenvector of $T$ with eigenvalue $\lambda_S = p^{|V(S)|}$, where $V(S)\subseteq [n]$ is the set of all vertices belonging to some edge in $S$.
\end{fact}
In particular we see that $M$ has a large spectral gap; its largest nontrivial eigenvalue is $\lambda_2:=\max_{S\neq \emptyset}|\lambda_S| = p^2 = n^{2(\alpha-\beta)}=o(1)$. Note that since $\E_\nul[A]=0$, $A$ is orthogonal to the constant function $\mathbf{1}$, and lies completely in the nontrivial eigenspace of $T$. This implies the following contractive inequality:
\begin{equation*}\label{eqn:f-contraction}
    \|TA\|_{2,\nul}\leq p^2\cdot \|A\|_{2,\nul} = p^2.
\end{equation*}
This suggests that we might want to set $g=TA$. Let us compute
\begin{equation}\label{eqn:Tf-bound}
    \adv^{(\pla,\nul)}(TA) = \frac{\E_\pla[TA]}{\|TA\|_{2,\nul}} = \frac{\E_{\pla^*}[A]}{\|TA\|_{2,\nul}}\geq \frac{\E_{\pla^*}[A]}{p^2}=\frac{\algadv^{(\pla^*,\nul)}(A)}{p^2}\gtrsim \frac{\adv^{(\pla^*,\nul)}(\F_{\leq d})}{p^2}.
\end{equation}
The second equality above uses the definition of $T$ and the fact that $M(\pla)=\pla^*$, and the final equality uses our simplifying assumption that $\E_\nul[A]=0$. A simple calculation (for instance, the one in \cref{sec:a3}) shows that the numerator is approximately $n^{-2\beta}$ for any constant $d$. On the other hand, the denominator is $p^2 = n^{2(\alpha-\beta)}$ . Consequently, we have shown
\[\adv^{(\pla,\nul)}(TA) \gtrsim \Omega(n^{-2\alpha}).\]
Unfortunately this is not strong enough; recall that we wanted to compute $g$ achieving $\adv^{(\pla,\nul)}(g) = \omega(1)$. So we must look beyond the spectral gap to achieve the desired amplification. One might imagine that if $A$ was orthogonal not only to $\mathbf{1}$, but also to \emph{all} eigenvectors corresponding to eigenvalues that are too large, $TA$ would obtain much larger advantage.

Our main idea is to achieve this by explicitly projecting $A$ onto the orthogonal complement of $\F_{\leq d} = \Span\{\chi_S : |S|\leq d\}$. Let us write $A=A_{+}+A_{-}$ for some low degree function $A_{+} \in \F_{\leq d}$ and some high degree function $A_{-}\perp \F_{\leq d}$, and imagine replacing $A$ with $A_-$ in the prior discussion. Observe that for every set $S$ of edges with $|S|> d$, we have $|V(S)|\geq \sqrt{|S|}>\sqrt{d}$, so $\lambda_S= p^{2 |V(S)|}< p^{2\sqrt{d}}$. As a consequence, $T$ acts as a $p^{2\sqrt{d}}$-contraction on $A_-$, an improvement from the $p^2$ before. 

Next, we will lower bound $\adv^{(\pla^*,\nul)}(A_-)$. Since $A_-$ and $A_+$ are orthogonal, $\max(\|A_-\|_{2,\nul},\|A_+\|_{2,\nul})\leq \|A\|_{2,\nul}= 1$, so one can show the following kind of ``triangle inequality'':
\[\algadv^{(\pla^*,\nul)}(A) = \E_{\pla^*}[A]=\E_{\pla^*}[A_+]+ \E_{\pla^*}[A_-]\leq \adv^{(\pla^*,\nul)}(A_{+})+\adv^{(\pla^*,\nul)}(A_{-}).\]
Recall that $\E_{\pla^*}[A] \geq (1+\Omega(1))\cdot \adv^{(\pla^*,\nul)}[\F_{\leq d}]$ by \cref{eqn:a-adv-bound}, and note that $\adv^{(\pla^*,\nul)}(A_+) \leq \adv^{(\pla^*,\nul)}[\F_{\leq d}]$ because $A_+\in \F_{\leq d}$. Together, these imply that $A_-$ achieves a nontrivial ``advantage'' of
\[\adv^{(\pla^*,\nul)}(A_-)\geq \Omega\left(\adv^{(\pla^*,\nul)}[\F_{\leq d}]\right).\]
In a similar fashion to \cref{eqn:Tf-bound}, we can conclude that
\[\adv^{(\pla,\nul)}(TA_{-})\geq \frac{\adv^{(\pla^*,\nul)}(A_{-})}{p^{2\sqrt{d}}} = \Omega\left(\frac{\adv^{(\pla^*,\nul)}[\F_{\leq d}]}{p^{2\sqrt{d}}}\right)\approx \Omega\left(n^{-2\beta}\cdot n^{2\sqrt{d}(\beta-\alpha)}\right).\]
Picking $d$ to be a constant greater than $\left(\frac{\beta}{\beta-\alpha}\right)^2$, we get $\adv^{(\pla,\nul)}(TA_{-})=\omega(1)$, as desired. Therefore we will use $g:=TA_-$.

\begin{remark}
    $TA_{-}(G)$ is not directly computable, but we can approximate it sufficiently well by making a polynomial number of queries to $A$.
\end{remark}

\paragraph{Step 2: Rounding.}

In this step, we show how to round $g$ to our final distinguisher $B$ for $\pla$ and $\nul$ achieving advantage $\Omega(1)$. We will construct $B$ by thresholding the value of $g=TA_-$. Concretely, we will compute $B$ in the following way.

\begin{itemize}
    \item On input $G\in \Omega$, compute $g(G)$.
    \item If $|g(G)|\geq \|g\|_{2,\pla}/10$, return $1$. Otherwise, return $-1$.
\end{itemize}

We will interpret an output of $1$ (resp. $-1$) as the algorithm guessing that the input was sampled from $\pla$ (resp. $\nul$). It is quite straightforward to bound the success probability under $\nul$; by Markov's inequality, one can show 
\[\Pr[|g(G)|\geq \|g\|_{2,\pla}/10]\leq 100\cdot \frac{\|g\|_{2,\nul}}{\|g\|_{2,\pla}}\leq 100\cdot \frac{\|g\|_{2,\nul}}{|\E_\pla[g]|}=\frac{100}{|\adv^{(\pla,\nul)}(g)|}=o(1).\]

The main challenge is to bound the probability of success under $\pla$. We will aim prove that $g$ is, say, $(\Omega(1),0.1)$-anticoncentrated (\cref{def:anticoncentration}) under $\pla$, that is, $|g(G)|\geq \|g\|_{2,\pla}/10$ with probability at least $\Omega(1)$ for $G\sim \pla$. If this is true, then $B$ achieves the requisite advantage
\[\algadv^{(\pla,\nul)}[B] = 2\cdot (\Pr_{\pla}[B(G)=1]-\Pr_{\nul}[B(G)=1]) = \Omega(1)-o(1) = \Omega(1).\]
So it suffices to prove that $g(\pla)=TA_-(\pla)$ is anticoncentrated. Our key insight here is that $A_-$ is a \emph{noise-resistant} function under $\pla$ with respect to the noise operator $T$. Concretely, one can prove
\[\frac{\|TA_-\|_{2,\pla}}{\|A_-\|_{2,\pla^*}}\gtrsim n^{-O(d)}.\]
Since $\pla$ is a relatively simple distribution, one should expect such noise resistant functions to behave like low-degree functions do over the boolean hypercube, in particular, it should be anticoncentrated. We prove the following anticoncentration statement for $Tf$ such noise-resistant functions $f$.

\begin{lemma}[\cref{lemma:planted-clique-hyp}, simplified version]\label{cref:simplified-anticoncentration}
    Suppose $f\in L^2(\pla^*)$ satisfies $\|Tf\|_{2,\pla}\geq n^{-(\beta-\alpha)D}\cdot \|f\|_{2,\pla^*}$ for some $D=o(\log n)$. Suppose that $f$ is invariant under permutating the vertices $[n]$. Then, $Tf(\pla)$ is $(c^{D^2},0.1)$-anticoncentrated for some universal constant $c$.
\end{lemma}
The proof of of the above lemma is quite technical so we will avoid explaining it here, pointing the reader to \cref{sec:anti} for the proof. We will invoke the lemma with $f=A_-$. Note that we can achieve permutation-invariance by simply randomly permuting the inputs to $A_-$. As a result, we obtain that $TA_-$ is indeed $(\Omega(1), 0.1)$-anticoncentrated, completing the argument that $B$ achieves advantage $\Omega(1)$.

\begin{remark}
    In this overview we have not shown how to obtain \cref{thm:perturbation}, our result for perturbations of $\pla$. Our argument above generalizes straightforwardly to this setting, and we directly prove a version of the reduction for perturbations of $\pla$ in \cref{sec:reduction}.
\end{remark}

\subsection{Computing a Hardcore Distribution}\label{subsec:moments}

The goal of this part is to discuss the main ideas behind \cref{thm:intro-moment-match}, which states that one can slightly perturb the standard Planted Clique distribution to make it arbitrarily hard to distinguish from an \ER\ graph with respect to the set of low-degree polynomials $\F_{\leq d}$. Concretely, let $\alpha>0$, and define
\[\pla = G(n,1/2,n^{1/2-\alpha}),\quad \nul = G(n,1/2),\quad \Omega=\{\pm 1\}^{\binom{[n]}{2}}.\]
We will show that
\begin{equation}\label{eqn:ld-advantage-zero}
 \adv^{(\pla',\nul)}[\F_{\leq d}] = 0   
\end{equation}
for some distribution $\pla'$ satisfying $\|\frac{d\pla'}{d\pla}\|_\infty\leq 1+o(1)$.

We will construct $\pla'$ by computing \emph{acceptance probabilities} $p(G)\in [0,1]$ for each graph $G$. Let $G\sim \pla$, and sample $a\sim \ber(p(G))$. We set $\pla'$ to be the distribution on $G$ conditioned on $a=1$. That is, the density $d\pla'$ will be given by
\[d\pla'(G) = \frac{p(G)\cdot d\pla(G)}{\E_\pla[p]}.\]
Notice that for any function $f:\Omega\to \R$,
\[\E_{\pla'}[f] = \frac{\langle f, p\rangle_{\pla}}{\E_\pla[p]}.\]
Recall (\cref{claim:ldlr-expression}) that one can compute the low-degree advantage $\adv^{(\pla',\nul)}[\F_{\leq d}]$ by
\begin{equation}\label{eqn:ldlr-term}
   \adv^{(\pla',\nul)}[\F_{\leq d}] =\sqrt{\sum_{0<|S|\leq d}\E_{\pla'}[\chi_S]^2}=\frac{1}{\E_\pla[p]}\sqrt{\sum_{0<|S|\leq d}\langle \chi_S,p\rangle_\pla^2}. 
\end{equation}
We would like $p$ to satisfy exactly the following two conditions:
\begin{enumerate}
    \item\label[cond]{item:cond2} The acceptance probability over $\pla$ is close to $1$. In particular, $\E_\pla[p]\geq 1-o(1)$. This would imply that $\|\frac{d\pla'}{d\pla}\|_\infty = \max_{G}\frac{p(G)}{\E_\pla[p]}\leq \frac{1}{\E_\pla[p]}=1+o(1)$.
    \item\label[cond]{item:cond1} The low degree moments of $\pla'$ are equal to those of $\nul$, in the sense that $\langle \chi_S, p\rangle_{\pla}=0$ for all $0<|S|\leq d$. This would imply that $\adv^{(\pla',\nul)}[\F_{\leq d}]=0$ by \cref{eqn:ldlr-term}.
\end{enumerate}

Notice that both of these conditions are linear constraints on the vector $p:\Omega\to [0,1]$. Therefore, we can write our objective as the following linear program:

\begin{equation*}
    \text{LP} \equiv \max\ \E_{\pla}[p]\quad \text{s.t. }
    \begin{cases}
        \langle p, \chi_S\rangle_\pla = 0\quad \text{ for all }S\subseteq \Omega, 0<|V(S)|\leq d,\\
        p:\Omega\to [0,1].
    \end{cases}
\end{equation*}

Our goal is to show that the optimal value of this LP is close to $1$. By computing the dual with respect to the first set of constraints, one can reduce this to proving an inequality for low degree functions over $\pla$. In particular, one can show
\begin{equation}\label{eqn:lp-dual}
  \text{LP} = \inf_{\alpha}\E_{G\sim \pla}[\max(0, 1+\alpha(G))],  
\end{equation}
where the infimum is taken over all low degree functions $\alpha$ with zero constant coefficient, that is, $\alpha = \sum_{0<|S|\leq d}\hat{\alpha}_S\cdot \chi_S$ for some Fourier coefficients $\{\hat{\alpha}_S\}_{0<|S|\leq d}$. Fix an arbitrary function $\alpha$ of this form. If $\max(0,1+\alpha)$ was a low-degree polynomial, one would easily be able to use the fact that $\pla$ and $\nul$ are mildly indistinguishable with low-degree polynomials to bound
\[\E_{\pla}[\max(0, 1+\alpha)]\geq \E_{\nul}[\max(0, 1+\alpha)]-o(1)\geq \E_{\nul}[1+\alpha]-o(1)=1-o(1).\]
Unfortunately $\max(0,1+\alpha)$ is clearly not a low-degree polynomial. Still, it is a ``simple'' enough function of the low-degree polynomial $1+\alpha$ that we are able to prove exactly the above bound.

With this, we have shown that $\text{LP}\geq 1-o(1)$, implying the existence of a $p$ satisfying \cref{item:cond1,item:cond2}. We have not yet commented on how to efficiently compute $p$. As written, the linear program has exponentially many variables, but nevertheless it is possible to find an approximately optimal solution that achieves
\[\adv^{(\pla',\nul)}[\F_{\leq d}]\leq \delta\]
in time $\poly(n/\delta)$. The main idea is to find a smooth convex relaxation to the dual convex program (\cref{eqn:lp-dual}), and optimize over $\alpha$ using stochastic gradient descent. Once we have found an approximate maximizer, we can push it back through the duality argument to find an approximate minimizer $p$ of the LP. We leave the details of this construction to \cref{sec:hardcore}.

\section{Hardness Amplification from Subspace Hardness}\label{sec:reduction}
The purpose of this section is to abstract the key ideas of our reduction and present a proof of a general hardness amplification result. Consider the problem of efficiently distinguishing two sequences of distributions $\pla=\pla_n$ and $\nul=\nul_n$ over $\Omega=\Omega_n$. Let $M=M_n$ be a reversible Markov chain with stationary distribution $\nul$, and let $T:L^2(\nul)\to L^2(\nul)$ be its associated noise operator. Let $\pla'$ be a distribution that is close to $\pla$ in the sense that $\|\frac{d\pla'}{d\pla}\|_\infty$ is bounded, and set $\pla^*=M(\pla')$.

We will provide conditions on $\pla$, $\nul$, and $T$ under which weak hardness for distinguishing $\pla$ and $\nul$ implies strong hardness for distinguishing $\pla^*$ and $\nul$. The performance of our reduction depends primarily on a parameter $\eps>0$. A central object in our statement is a subspace $V\subseteq L^2(\nul)$ containing the top $\geq \eps$-eigenspace of $T$, that is, $\lambda_{>\eps}(T)\subseteq V$. $V$ is often interpretable as a space consisting of the ``simplest'' functions in $L^2(\nul)$.

\begin{example}[Low-degree functions]\label{example:low-degree}
    When $T=T_\rho$ is the standard noise operator on the boolean hypercube, $V=\F_{\leq d}$ can be taken as the set of degree $\leq d$ functions for $d= \log_\rho\eps$.
    
    When $T$ is the vertex-resampling noise operator with noise rate $1-p$, $V=\F_{\leq d}$ can be taken as the set of degree $\leq d$ functions for $d=\binom{\log_p\eps}{2}$ (by \cref{fact:eigen}).
\end{example}

We will require that $\eps$ is large enough so one can find such a $V$ that is \emph{computationally tractable}. Concretely, we impose the following set of conditions.

\begin{definition}[Tractability of function spaces]\label{def:tractable-subspace}
    Let $\nul=\nul_n$ be a sequence of distributions over $\Omega=\Omega_n$. A sequence of subspaces $V=V_n\subseteq L^2(\nul)$ containing the constant function $\mathbf{1}$ is $q(n)$-tractable if it has an orthonormal basis $\{\mathbf{1},f_1,\ldots, f_\ell\}$ such that
    \begin{enumerate}
        \item $\ell\leq q(n)$.
        \item For all $i$, one can evaluate $f_i$ in time $q(n)$.
        \item Functions $f$ in $V$ are $q(n)$-smooth, that is, $\|f\|_\infty\leq q(n)\cdot \|f\|_{2,\nul}$.
    \end{enumerate}
\end{definition}

We note that both examples in \cref{example:low-degree} are $n^{O(d)}$-tractable. We will also require a technical condition stating that functions $f\in L^2(\pla)$ that are noise-resistant should be well-behaved as random variables over $\pla$, in the sense that they are anticoncentrated (\cref{def:anticoncentration}). Finally, we assume that the problem of distinguishing $\pla$ from $\nul$ is mildly hard. Informally, the conclusion of our theorem states that if the above conditions hold, then for any efficiently computable test $A$,
\begin{equation}\label{eqn:informal-reduction}
  \algadv^{(\pla^*,\nul)}(A) < \adv^{(\pla^*,\nul)}[V] + \delta  
\end{equation}
for some small $\delta$.


\begin{theorem}\label{thm:general-1}
    Let $\nul=\nul_n$ and $\pla=\pla_n$ be two sequences of distributions over a finite sample space $\Omega=\Omega_n$. Let $M=M_n$ be a reversible Markov chain with stationary measure $\nul$, and let $T:L^2(\nul)\to L^2(\nul)$ be its associated noise operator. Let $q=q(n)\geq 1,\eps=\eps(n)\in [0,1]$, and $\gamma=\gamma(n)\in [0,1]$ be parameters, and let $V=V_n$ be a subspace of $L^2(\nul)$ containing the top eigenspace $\lambda_{>\eps}(T)$. Assume the following ``niceness'' conditions.
    \begin{enumerate}
        \item In time $q(n)$, one can sample $y\sim\nul$ and $y\sim M(x)$ for any $x\in \Omega$.\hfill(sampleability)
        \item $V$ is $q$-tractable.\hfill(tractability of top eigenspace)
        \item \label{cond:anti}If $\|Tf\|_{2,\pla}\geq \eps/q\cdot \|f\|_{2,M(\pla)}$, then $Tf(\pla)$ is $(\gamma,0.1)$-anticoncentrated.\hfill(anticoncentration)
    \end{enumerate}

    Let $\delta = 400 \gamma^{-1} \eps$. Let $\pla'=\pla'_n$ be another sequence of distributions such that $\|\frac{d\pla'}{d\pla}\|_\infty < \infty$, and set $\pla^*=M(\pla')$. Suppose one can compute a randomized mapping $A$ from $\Omega$ to $\{\pm 1\}$ in time $t(n)$ satisfying 
    \begin{equation}\label{eqn:reduction-assumption}
        \algadv^{(\pla^*,\nul)}(A)\geq \var_{\nul}(\Pi_V f)^{1/2}\cdot \adv^{(\pla^*,\nul)}[V] + \delta\cdot \left\|\frac{d\pla'}{d\pla}\right\|_\infty
    \end{equation}
    for infinitely many $n$, where $f(G) = \mathbb{E}_{A}[A(G)]$ and $\var_{\nul}(\Pi_V f) = \|\Pi_V f - \E_{\nul}[\Pi_V f]\|_{2,\nul}^2$. Then, there is a randomized mapping $B$ from $\Omega$ to $\{\pm 1\}$ with runtime $t(n)\cdot \poly(q,\eps^{-1},\gamma^{-1})$ such that for infinitely many $n$,
    \[\algadv^{(\pla,\nul)}(B)\geq \gamma.\]
\end{theorem}

\subsection{Proof of \cref{thm:general-1}}
     Assume without loss of generality that $\E_{\pla^*}[A]\geq \E_\nul[A]$, so $\algadv^{(\pla^*,\nul)}(A) = \E_{\pla^*}[A]-\E_\nul[A]$. We will define $f(x)=\E_{A}[A(x)]$ to be the average output of the algorithm under input $x$.
 
    Let us recall the rough plan of our reduction, as explained in \cref{sec:overview-reduction}. We will decompose $f = f_{+} + f_{-}$ for some $f_{+}\in V$ and $f_{-} \perp V$. We will obtain the test $B$ by thresholding $|T f_{-}|$.
    
    Of course, we aren't necessarily able to compute $Tf_{-}$ efficiently -- we are only guaranteed that $f$ is efficiently computable. We will instead compute an approximation that relies on the fact that $V$ is tractable.

    \begin{claim}\label{claim:approx-high-degree-part}
        There is a randomized mapping $C$ from $\Omega$ to $\R$ computable in time $t(n)\cdot \poly(q,\eps^{-1},\gamma^{-1})$ such that for any $x\in \Omega$,
        \begin{equation}\label{eqn:A-bound}
            \E_{C}[|C(x) - Tf_{-}(x)|] \leq \eps.
        \end{equation}
    \end{claim}
    We prove \cref{claim:approx-high-degree-part} in \cref{sec:a2}. Let us continue by defining our test $B$.
    \paragraph{Definition of our test $B$ for $\pla$ and $\nul$.}
    \begin{enumerate}
        \item On input $x$, compute $C(x)$.
        \item If $|C(x)|\geq\delta/20$, output $+1$. Otherwise output $-1$.
    \end{enumerate}

    In the rest of the proof, we will analyze the probability that this algorithm succeeds under $\nul$ and $\pla$ and conclude that it achieves distinguishing advantage $\Omega(\gamma)$.

    \paragraph{Probability of success under $\nul$.} 
    We will show that $|C(x)|$ is small with high probability by bounding its average and applying Markov's inequality.
    \begin{align*}
        \E_{x\sim \nul, C}[|C(x)|]&\leq \|Tf_{-}\|_{1,\nul} + \eps\tag{triangle inequality, \cref{eqn:A-bound}}\\
        &\leq \|Tf_{-}\|_{2,\nul} + \eps\tag{Jensen's inequality on $t\to t^2$}\\
        &\leq \eps \cdot\|f_{-}\|_{2,\nul} + \eps\tag{$V^\perp$ is contained in the $<\eps$-eigenspace of $T$}\\
        &\leq 2\eps.\tag{$\|f_{-}\|_{2,\nul}\leq \|f\|_{2,\nul}\leq \|f\|_\infty\leq 1$}
    \end{align*}
    For the last inequality, we used that the outputs of $f$ are bounded in $[-1,1]$ because $A$ has range $\{\pm 1\}$. By Markov's inequality, the probability that $|C(x)|\geq \delta/20$ is at most $400\eps/\delta$. Plugging in the definition of $\delta$, this is at most $\gamma/10$. Therefore, $B$ outputs $+1$ with probability at most $\gamma/10$.

    \paragraph{Probability of success under $\pla$.}

    We start by showing that the anticoncentration condition (condition \ref{cond:anti}) can be applied to $Tf_{-}$, i.e.
    $\|Tf_{-}\|_{2,\pla}\geq \eps/q\cdot \|f_{-}\|_{2,M(\pla)}$. Let us bound
    \begin{align*}
        \|Tf_{-}\|_{2,\pla}&\geq \|Tf_{-}\|_{1,\pla}\tag{Jensen's inequality on $t\to t^2$}\\
        &\geq \left\|\frac{d\pla'}{d\pla}\right\|_\infty^{-1}\cdot \|Tf_{-}\|_{1,\pla'}\tag{$\|\frac{dP}{dQ}\|_\infty = \sup_{f}\frac{\|f\|_{1,P}}{\|f\|_{1,Q}}$}\\
        &\geq \left\|\frac{d\pla'}{d\pla}\right\|_\infty^{-1}\cdot \E_{\pla'}[Tf_{-}]\\
        &=\left\|\frac{d\pla'}{d\pla}\right\|_\infty^{-1}\cdot \E_{\pla^*}[f_{-}].
    \end{align*}
    \cref{eqn:reduction-assumption} will allow us to further lower bound this quantity. Write $\delta' = \delta\cdot\|\frac{d\pla'}{d\pla}\|_\infty$. For infinitely many $n$, we have
    \begin{align*}
        \E_{\pla^*}[f_{-}] &= \E_{\pla^*}[f] - \E_{\pla^*}[f_{+}]\\
        &= \E_{\pla^*}[A] - \E_{\pla^*}[f_{+}]\tag{definition of $f$}\\
        &\geq \E_{\nul}[A] + \var_\nul(f_+)^{1/2}\cdot \adv^{(\pla^*,\nul)}[V] + \delta' - \E_{\pla^*}[f_{+}]\tag{\cref{eqn:reduction-assumption}}\\
        &\geq \E_{\nul}[f] + \var_\nul(f_+)^{1/2}\cdot \adv^{(\pla^*,\nul)}[V] + \delta' - \E_{\nul}[f_{+}] - \adv^{(\pla^*,\nul)}[V]\cdot \var_\nul(f_+)^{1/2}\tag{$f_{+}-\E_{\nul}[f_{+}]\in V$}\\
        &= \E_{\nul}[f] + \delta' - \E_{\nul}[f_{+}]\\
        &= \E_{\nul}[f_{-}]+\delta'=\delta'.\tag{$f_{-}\perp V$ and $\mathbf{1}\in V$}
    \end{align*}
    To elaborate on the fourth step, we applied the definition of $\adv^{(\pla^*,\nul)}[V]$ on the function $g=f_+ - \E_{\nul}[f_+]$, which also belongs to $V$ because $V$ contains the constant function $\mathbf{1}$. The inequality then follows by the fact $\|g\|_{2,\nul}^2 =  \|f_{+} - \mathbb{E}_\nul[f_{+}]\|_{2,\nul}^2= \var_\nul(f_{+})$.
    
    Combining the two inequalities proved above, we conclude that $\|Tf_{-}\|_{2,\pla}\geq  \delta\geq 2\eps$. On the other hand,
    \begin{align*}
        \|f_{-}\|_{2,M(\pla)}&\leq \|f\|_{2,M(\pla)} + \|f_{+}\|_{2,M(\pla)}\tag{triangle inequality}\\
        &\leq \|f\|_{\infty} + \|f_{+}\|_{\infty}\\
        &\leq q\|f\|_{2,\nul} + q\|f_+\|_{2,\nul}\tag{smoothness of $V$ (\cref{def:tractable-subspace})}\\
        &\leq 2q.\tag{$\|f_+\|_{2,\nul}\leq \|f\|_{2,\nul}\leq 1$}
    \end{align*}

    Therefore we have $\|Tf_{-}\|_{2,\pla}\geq \eps/q\cdot \|f_{-}\|_{2,M(\pla)}$, and we can conclude that $Tf_{-}(\pla)$ is $(\gamma, 0.1)$-anticoncentrated. Finally, we can bound the probability that $B$ outputs +1. By a union bound, we write
    \[\Pr_{\pla,C}[|C(x)|\geq \delta/20]\geq \Pr_{x\sim \pla}[|Tf_{-}(x)|\geq \delta/10] - \Pr_{x\sim \pla, C}[|C(x)-Tf_{-}(x)|\geq \delta/20].\]
    By $(\gamma,0.1)$-anticoncentration, the first term is at least $\gamma$. Using Markov's inequality along with \cref{eqn:A-bound}, the second term is at most $20\eps/\delta\leq \gamma/10$, so the probability above can be bounded by $\gamma-\gamma/10=9\gamma/10$.

    Therefore, $B$ outputs $+1$ with probability at least $9\gamma/10$ under $\pla$. Combining our bounds on the success probability in both cases, we get that $B$ achieves the required advantage
    \[\algadv^{(\pla,\nul)}(B) = 2\cdot (\Pr_{x\sim \pla}[B(x)=1] - \Pr_{x\sim \nul}[B(x)=-1])\geq 2 \cdot (9\gamma/10 - \gamma/10) \geq\gamma\]
    for infinitely many $n$.

\section{Hardcore Distributions against a Subspace of Distinguishers}\label{sec:hardcore}
In this section, we will state and prove a sort of hard-core Lemma for distinguishing two distributions $\pla$ and $\nul$ over a sample space $\Omega$ using functions belonging to a computationally tractable subspace $V\subseteq L^2(\nul)$. Assume that $V$ contains the constant function $\mathbf{1}$. Recall our measure of performance of a $V$ at the task:
\[\adv^{(\pla,\nul)}[V] = \max_{f\in V}\adv^{(\pla,\nul)}(f) = \max_{f\in V}\frac{\E_\pla[f]-\E_\nul[f]}{\|f\|_{2,\nul}}.\]

Our result states that if $V$ is nice enough and $\pla$ satisfies a mild hardness condition against $V^{\odot 4}$, then one can ``perturb'' $\pla$ very slightly to find another distribution $\pla^*$ that satisfies arbitrarily strong hardness against $V$.

\begin{theorem}\label{thm:general-2}
    Let $\pla=\pla_n$ and $\nul=\nul_n$ be distributions over a finite set $\Omega=\Omega_n$, and let $V=V_n\subseteq L^2(\nul)$ be a subspace of functions containing the constant function $\mathbf{1}$. Let $q=q(n)\geq 1$ and $c=c(n)\geq 1$ be parameters. Assume the following conditions.
    \begin{enumerate}
        \item One can sample from $\pla$ in time $q(n)$.\hfill (efficient sampleability)
        \item $V$ is $q$-tractable.\hfill (tractability of $V$)
        \item Any $f\in V$ satisfies $\|f\|_{8,\nul}\leq c^{1/4}\cdot \|f\|_{4,\nul}$\hfill ($V$ is $(4,8)$-hypercontractive)
        \item $\adv^{(\pla,\nul)}[V^{\odot 4}]\leq 1/(8c)$.\hfill (mild hardness assumption)
    \end{enumerate}
    For any $\delta=\delta(n)>0$, one can using $\poly(q,\delta^{-1})$ preprocessing time prepare the description of a distribution $\pla^*$ such that with high probability,
    \begin{enumerate}
        \item\label{concl:1} $\left\|\frac{d\pla^*}{d\pla}\right\|_\infty\leq 1+O\left(c\cdot \adv^{(\pla,\nul)}[V]\right)$\hfill(closeness of $\pla^*$ and $\pla$)
        \item\label{concl:2} One can sample from $\pla^*$ in expected time $\poly(q,\log(\delta^{-1}))$.\hfill (efficient sampleability)
        \item $\adv^{(\pla^*,\nul)}[V]\leq \delta$.\hfill (hardness amplified to $\delta$)
    \end{enumerate}
\end{theorem}

As a concrete example, one can consider the setting of the boolean hypercube, where $\nul = \unif(\{\pm 1\}^n)$ and $V=\F_{\leq d}$ is the space of degree $\leq d$ functions. In this case, the relevant parameters can be shown to be $q(n) = n^{d}$ and $c(n) = 2^{O(d)}$.

We make a few remarks about the theorem statement.

\begin{itemize}
    \item The bound on $\left\|\frac{d\pla^*}{d\pla}\right\|_\infty$ in the conclusion is much stronger than a TV distance bound -- in particular, $\pla^*$ is contiguous with respect to $\pla$, that is, unlikely events in $\pla$ remain unlikely in $\pla^*$. This fact is false for two distributions that are $\delta$-close in TV distance.
    \item Under the conditions of the theorem, one can show using a compactness argument that there exists a distribution $\pla^*$ satisfying $\left\|\frac{d\pla^*}{d\pla}\right\|_\infty\leq 1+O\left(c\cdot \adv^{(\pla,\nul)}[V]\right)$ and $\adv^{(\pla^*,\nul)}[V] = 0$. That is, functions in $V$ achieve advantage exactly $0$.
\end{itemize}

\subsection{Proof of \cref{thm:general-2}}

Let us define $\eps = \adv^{(\pla,\nul)}[V^{\odot 4}]$. We will assume $\delta\leq \eps$, since otherwise $\pla^*=\pla$ already satisfies the requirement that $\adv^{(\pla^*,\nul)}[V]\leq \delta$. As explained in \cref{subsec:moments}, we will construct $\pla^*$ by computing \emph{acceptance probabilities} $p:\Omega\to [0,1]$. We set $\pla^*$ to be $\pla$ conditioned on the event that $\ber(p(x))=1$. In other words,
\begin{equation}\label{eqn:formula-phi}
  d\pla^*(x) = \frac{d\pla(x)\cdot p(x)}{\E_\pla[p]}.  
\end{equation}
We will sample from $\pla^*$ by performing rejection sampling with $\pla$. 

Recall that in \cref{subsec:moments}, we defined $p$ to be the optimal solution to a particular LP. As solving the LP directly takes exponential time, we will follow the outline from the end of the technical overview and attempt to solve a \emph{smoothing} of the dual formulation. Recall the dual formulation:
\[\min_{g=\sum_{i\in [\ell]}g_i f_i}\E_{x\sim \pla}[\max(0,1+g(x))].\]
Let us define a smooth approximation to the univariate function $t\to \max(0, t)$; let $\sigma(t) = \int_{-\infty}^t\frac{1}{1+2^{-x/\delta}}dx$. Observe that $\sigma$ is convex and continuously differentiable, and satisfies the following properties for all $t\in \R$.
\begin{enumerate}
    \item $\sigma'(t)\in [0,1]$.
    \item $\sigma(t)\geq \max(0, t)$.
    \item $\sigma(0)\leq \delta$.
    \item $\sigma''(t)\leq O(1/\delta)$.
\end{enumerate}

Let $\{\mathbf{1}, f_1,\ldots, f_\ell\}$ be a computationally tractable orthonormal basis of $V$, where $\mathbf{1}$ is the constant function. Our algorithm will consist of finding the optimal solution to the following convex program:
\[\min_{g=\sum_{i\in [\ell]}g_i f_i}\E_{x\sim \pla}[\sigma(1+g(x))].\]
Say we find a solution $g^*$ to the convex program. We will set the acceptance probabilities to 
\begin{equation}\label{eqn:p-defn}
    p(x) := \sigma'(1+g^*(x)).
\end{equation}

One cannot expect to find an exact optimizer in general, since even exactly evaluating the objective function is intractable. We will leverage the fact that the objective function is written as an average over efficiently computable functions, and use stochastic gradient descent to find an approximate optimizer. We will only require approximate first-order optimality in the sense that the gradient of the objective function at $g^*$ is small. Let us calculate the first derivatives.
\[\left.\frac{\partial \E_{\pla}[\sigma(1+g)]}{\partial g_i}\right|_{g=g^*} = \left.\E_{x\sim \pla}\left[\frac{\partial}{\partial g_i} \sigma\left(1+\sum_{i\in [\ell]}g_i f_i\right)\right|_{g=g^*}\right] = \E_{x\sim \pla}\left[f_i\cdot \sigma'(1+g^*)\right] = \langle f_i, p\rangle_\pla.\]
\begin{claim}\label{claim:sgd}
    With $\poly(q,\delta^{-1})$ preprocessing time, one can prepare the description of a function $g^*=\sum_{i\in [\ell]}g_i^* f_i$ that can be queried in time $\poly(q, \log \delta^{-1})$ such that with high probability over the preprocessing randomness,
    \[\left.\left\|\nabla_{\{g_1,\ldots, g_\ell\}}\E_{\pla}[\sigma(1+g)]\right\|_{2} \right|_{g=g^*}= \sqrt{\sum_{i\in [\ell]}\langle f_i, p\rangle^2_{\pla}}\leq \delta/3,\]
    where $\left.\left\|\nabla_{\{g_1,\ldots, g_\ell\}}\E_{x\sim \pla}[\sigma(1+g(x))]\right\|_{2} \right|_{g=g^*}$ is the gradient of $\E_{x\sim \pla}[\sigma(1+g(x))]$ viewed as a function of the coefficients $g_1,\ldots, g_n$ of $g=\sum_{i\in \ell}g_i f_i$, evaluated at $g^*$.
\end{claim}

The proof of \cref{claim:sgd} involves invoking some off-the-shelf convergence guarantees for stochastic gradient descent, and we will defer it to \cref{sec:a1}. We will require another claim that bounds the optimal value of the convex program above.

\begin{claim}\label{claim:optim-bound}
    For any $g = \sum_{i\in [\ell]}g_i f_i$, we have $\E_{\pla}[\sigma(1+g)]\geq 1-4c \eps + \eps \|g\|_{2,\nul}$.
\end{claim}

The proof of \cref{claim:optim-bound} involves using the fact that $\pla$ and $\nul$ are mildly indistinguishable with respect to polynomials over $V^{\odot 4}$, which allows us to switch the expectation with respect to $\pla$ to an expectation with respect to $\nul$, which is then easily bounded. Let us argue that the three conditions on $\pla^*$ are satisfied.

\paragraph{Closeness of $\pla^*$ and $\pla$.} We will bound
\[\left\|\frac{d\pla^*}{d\pla}\right\|_\infty = \max_{x\in \Omega}\frac{p(x)}{\E_\pla[p]}\leq \frac{1}{\E_\pla[p]}.\]
Let us lower bound the denominator.
\begin{align*}
    \E_\pla[p] &= \E_\pla[\sigma'(1+g^*)]\tag{\cref{eqn:p-defn}}\\
    &= \E_\pla[(1+g^*)\cdot \sigma'(1+g^*)] - \E_\pla[g^*\cdot \sigma'(1+g^*)]\tag{adding/subtracting $\E_\pla[g^*\cdot \sigma'(1+g^*)]$}\\
    &\geq \E_\pla[\sigma(1+g^*)]- \sigma(0) - \E_\pla[g^*\cdot \sigma'(1+g^*)]\tag{$\sigma$ is convex, so $\sigma(t) - \sigma(0)\leq t\cdot \sigma'(t)$}\\
    &\geq 1-5c\eps + \eps \|g^*\|_{2,\nul} -  \E_\pla[g^*\cdot \sigma'(1+g^*)]\tag{\cref{claim:optim-bound}, $\sigma(0)\leq \delta\leq \eps\leq c\eps$}\\
    &\geq 1-5c\eps + \eps\|g^*\|_{2,\nul} - \sum_{i\in [\ell]}g_i^*\cdot \E_\pla[f_i\cdot p]\tag{expanded $g^*$, \cref{eqn:p-defn}}\\
    &\geq 1-5c\eps + \eps\|g^*\|_{2,\nul} - \|g^*\|_{2,\nul}\cdot\sqrt{\sum_{i\in [\ell]}\langle f_i,p\rangle_\pla^2}\tag{Cauchy-Schwarz}\\
    &\geq 1-5c\eps + \eps\|g^*\|_{2,\nul} - \delta/3\cdot \|g^*\|_{2,\nul}
    \geq 1-5c\eps.\tag{\cref{claim:sgd}, $\delta\leq \eps$}
\end{align*}
By the mild hardness assumption on $V^{\odot 4}$, $c\eps\leq 1/8$, implying that $\|\frac{d\pla^*}{d\nul}\|_\infty = 1+O(c\cdot \eps)$.

\paragraph{Efficient Sampleability}

We will sample from $\pla^*$ by rejection sampling with $\pla$. Concretely, we sample $x\sim \pla$, and a Bernoulli $b\sim \ber(p(x))$. If $b=1$, we output $x$, otherwise we repeat this process.
The output is distributed as $\pla^*$ by definition, and since $\E[p(x)]\geq 1-5c\eps > \Omega(1)$, the expected number of iterations before terminating is $O(1)$. Since $g^*$ can be computed in time $\poly(q,\log \delta^{-1})$, the expected runtime is $\poly(q,\log \delta^{-1})$.

\paragraph{Hardness amplified to $\delta$.} Finally, we bound the quantity $\adv^{(\pla^*,\nul)}[V]$ using the expression given by \cref{claim:ldlr-expression}.

\begin{align*}
    \adv^{(\pla^*,\nul)}[V] &=\sqrt{\sum_{i\in [\ell]}\E_{\pla^*}[f_i]^2}\tag{\cref{claim:ldlr-expression}}\\
    &= \frac{1}{\E_{\pla}[p]}\sqrt{\sum_{i\in [\ell]}\langle p, f_i\rangle_{\pla}^2}\tag{\cref{eqn:formula-phi}}\\
    &\leq \frac{\delta}{3\cdot (1-5c\eps)}\leq \delta.\tag{\cref{claim:sgd}, bound on $\E_\pla[p]$, $c\eps\leq 1/8$}
\end{align*}
This completes the proof of \cref{thm:general-2} modulo \cref{claim:sgd} which we prove in in \cref{sec:a1}, and \cref{claim:optim-bound} which we prove below.

\begin{proof}[Proof of \cref{claim:optim-bound}]
    Recall the hypercontractivity property we assumed, which states that for $f\in V$, $\|f\|_{8,\nul}\leq c^{1/4}\cdot \|f\|_{4,\nul}$. We will use this in several ways throughout this proof, so it will be helpful to collect all the consequences we will need before getting into the proof. Firstly, for any function $f\in V$, $(4,8)$-hypercontractivity implies $(2,4)$-hypercontractivity by \cref{lemma:log-convexity}:
    \begin{equation}\label{eqn:hyper-1}
        \|f\|_{4,\nul}\leq c^{1/2}\cdot \|f\|_{2,\nul}.
    \end{equation}
    As a result, one can apply the definition of $\eps=\adv^{(\pla,\nul)}[V^{\odot 4}]$ on the functions $f^2$ and $f^4$ to bound $|\|f\|_{2,\pla}^2-\|f\|_{2,\nul}^2|\leq \eps\cdot \|f\|_{4,\nul}^2$, and $|\|f\|_{4,\pla}^4-\|f\|_{4,\nul}^4|\leq \eps\cdot \|f\|_{8,\nul}^2$. Applying $(2,4)$-hypercontractivity and $(4,8)$-hypercontractivity respectively, one has the following inequalities.
    \begin{equation}\label{eqn:norm-preserved}
        \|f\|_{2,\pla}\geq \sqrt{1-c\eps}\cdot \|f\|_{2,\nul}\geq 2^{-1/5}\cdot \|f\|_{2,\nul},\quad \|f\|_{4,\pla}\leq (1+c^{1/2}\eps)^{1/4}\cdot \|f\|_{4,\nul}\leq 2^{1/5}\cdot \|f\|_{2,\nul},
    \end{equation}
    where we used the bound $c\eps \leq 1/8$. Let us finally begin with the proof. Recall that we want to prove a lower bound on $\E_\pla[\sigma(1+g)]$. Let $0<t<1/2$ be some parameter to be picked later. We will bound
    \begin{align*}
        \E_\pla[\sigma(1+g)] &\geq \E_\pla[\max(0, 1+g)]\tag{$\sigma(x)\geq \max(0, x)$}\\
        &= \frac{1}{2}\E_\pla[1+g] + \frac{1}{2}\E_\pla[|1+g|]\tag{$\max(0, x) = x/2 + |x|/2$}\\
        &\geq (1-t)\cdot\E_\pla[1+g] + t\cdot\E_\pla[|1+g|]\tag{$\E_\pla[|1+g|]\geq \E_\pla[1+g]$}\\
        &= 1-t + (1-t)\cdot \E_\pla[g] + t\cdot \|1+g\|_{1, \pla}.
    \end{align*}
    We will bound the second and third terms separately. For the second term, recall that $\eps = \adv^{(\pla,\nul)}(V^{\odot 4})$ and $\E_\nul[g]=0$ because $g$ has no $\mathbf{1}$ component. So,
    \[\E_\pla[g] \geq \E_\nul[g]-\eps\cdot \|g\|_{2,\nul}= -\eps\cdot \|g\|_{2,\nul}.\]
    Next we bound the third term.
    \begin{align*}
        \|1+g\|_{1, \pla}&\geq  \frac{\|1+g\|_{2,\pla}^{3}}{\|1+g\|_{4,\pla}^2}\tag{\cref{lemma:log-convexity}}\\
        &\geq \frac{1}{2}\cdot \frac{\|1+g\|_{2,\nul}^{3}}{\|1+g\|_{4,\nul}^2}\tag{\cref{eqn:norm-preserved} on $f=1+g$}\\
        &\geq (2c)^{-1} \cdot \|1+g\|_{2,\nul}\geq (2c)^{-1}\cdot \|g\|_{2,\nul}.\tag{\cref{eqn:hyper-1} on $f=1+g$}
    \end{align*}
    Plugging these back into our original bound, we get
    \begin{align*}
        \E_\pla[\sigma(1+g)] &\geq 1-t + \|g\|_{2,\nul}\cdot \left(\frac{t}{2c} - \eps\right)\\
        &= 1-4\cdot c \eps + \eps \|g\|_{2,\nul},\tag{picking $t=4 \eps c\leq 1/2$.}
    \end{align*}
    completing the proof.
\end{proof}

\section{Anticoncentration for Planted Clique}\label{sec:anti}

In this section, we will prove the main technical statement that is required to invoke our general results in the setting of Planted Clique. We will show that Condition \ref{cond:anti} in the statement of \cref{thm:general-1} holds, that is, functions in $L^2(\pla)$ that survive the noise operator are anticoncentrated as real-valued random variables.

We will diverge slightly from the technical overview in \cref{sec:overview} and consider a slight modification of the vertex resampling Markov Chain introduced in \cref{sec:overview-reduction}. 

\begin{definition}[Permuted Vertex Resampling Markov Chain]
    Let $\Omega=\{\pm 1\}^{\binom{[n]}{2}}$. The Permuted Vertex Resampling Markov Chain with noise rate $1-p$ is a Markov chain $M^{\sym}$ with the following transition rule: for any $G\in \Omega$, $M^{\sym}(G)$ is distributed as $H$ in the following.
    \begin{itemize}
        \item Sample a random permutation $\pi\sim S_n$. Replace $G$ by $\pi(G)$.
        \item Sample $z\sim \ber(p)^{\otimes n}$.
        \item Sample the entries of $H$ as follows:
        \[H_{\{i,j\}} \sim \begin{cases}
            G_{\{i,j\}}\quad &\text{if }x_i=x_j=1,\\
            \rad\quad&\text{otherwise}.
        \end{cases}\]
    \end{itemize}
\end{definition}

The required anticoncentration statement is captured in the following Lemma.

\begin{lemma}\label{lemma:planted-clique-hyp}
    Let $0<p<1$ and $k\in \N$. Let $\pla = G(n, 1/2, k)$, $\pla'=G(n, 1/2, pk)$. Let $M^{\sym}$ be the permuted vertex resampling Markov chain over $\{\pm 1\}^{\binom{[n]}{2}}$ with noise rate $1-p$, so $M(\pla) = \pla'$. Let $T$ be the noise operator associated with $M^{\sym}$. Suppose $f\in L^2(\pla')$ satisfies $\|Tf\|_{2,\pla} \geq p^{d} \cdot\|f\|_{2,\pla'}$ for some $d=o(\min(k,\log p^{-1}))$. Then, $Tf(\pla)$ is $(c^{d^2},0.1)$-anticoncentrated for some universal constant $c$.
\end{lemma}

The goal of the rest of this section is to prove the above Lemma. Let us begin by describing a high-level plan. Observe that since $\pla$ is not a stationary distribution of $M^{\sym}$, we are not guaranteed a basis of eigenvectors of $T$ that is orthonormal with respect to $\pla$. To address this issue, consider the following equivalent way to sample a graph $\tilde{G}\sim \pla$. Define $q:=k/n$.
\begin{itemize}
    \item Sample $x\sim \ber(q)^{\otimes n}$. Define $\clique(x)\in \{0,1\}^{\binom{[n]}{2}}$ by $\clique(x)_{\{i,j\}}=x_ix_j$.
    \item Sample $G\sim G(n,1/2)$.
    \item Output $\tilde{G} = G\wedge \clique(x)$. That is, output the graph obtained from $G$ by planting a clique on the vertices $i$ such that $x_i=1$.
\end{itemize}
The main idea is to treat $f$ not as a function of $\tilde{G}$ under $\tilde{G}\sim \pla$, but as a function of $(x,G)$ where $x\sim \ber(q)^{\otimes n}$ and $G\sim G(n,1/2)$ are independent. Since the input is now a simple product distribution, we will be able to better understand the effect of the noise operator $T$ on it. We will begin by defining $\Psi$ and $\Psi'$, the ``lifted'' versions of $\pla$ and $\pla'$, by
\[\Psi=\ber(p)^{\otimes n}\times G(n,1/2),\qquad \Psi'=\ber(pq)^{\otimes n}\times G(n,1/2).\]
Note that $\Psi$ and $\Psi'$ are both distributions over $\{0,1\}^{[n]}\times \{\pm 1\}^{\binom{[n]}{2}}$. For clarity, we will always denote a sample from $\Psi$ as $(x,G)$ and a sample from $\Psi'$ as $(y,H)$ for $x,y\in \{0,1\}^{[n]}$ and $G,H\in \{\pm 1\}^{\binom{[n]}{2}}$. We define a lifted Markov chain $M^*$ by the following sampling proceduer for $M^*(x,G)$ for any $(x,G)$:
\begin{itemize}
    \item Randomly permute the vertices by sampling $\pi\sim S_n$, and replacing $x$ with $\pi(x)$ and $G$ with $\pi(G)$.
    \item Sample $z\sim \ber(p)^{\otimes n}$. Set $y=x\wedge z$.
    \item Set $H_{ij} = \begin{cases}
        G_{ij}\quad&\text{if }z_i=z_j=1,\\
        \rad\quad &\text{otherwise}.
    \end{cases}$
    \item Output $(y,H)$.
\end{itemize}
Let us denote by $T^*$ the noise operator associated with $M^*$. Note that $T^*(\Psi) = \Psi'$, and for any $(x,G)$, the distribution of $H\wedge \clique(y)$ for $(y,H)\sim M^*(x,G)$ is identical to that of $M^{\sym}(G\wedge \clique(y))$. As a consequence, we have the following.
\begin{observation}\label{obs:equiv}
    For any function $f\in L^2(\pla')$, consider $g(y,H)=f(H\wedge \clique(y))$. Then, $f(\pla')$ and $g(\Phi')$ are identical as real-valued random variables, and so are $Tf(\pla)$ and $T^*g(\Phi)$.
\end{observation}

We will be able to coarsely understand the action of $T^*$ on $L^2(\Psi')$ in terms of the Fourier basis. In particular, we will ``merge'' all monomials belonging to a certain equivalence class, resulting in a partitioning on the space $L^2(\Psi)$ into a collection of orthogonal subspaces. We say that a pair $(A,B)$ for $A\subseteq \binom{[n]}{2}$ and $B\subseteq [n]$ is ``valid'' if $V(A)\cap B = 0$. For a valid pair $(A,B)$, define the subspace $W_{A,B}\subseteq L^2(\Psi)$ to be all functions of the form
\[(x,G)\to G_A\cdot \chi^{q}_B(x)\cdot r(x_{V(A)}),\]
where $r$ is an arbitrary function in $L^2(\{0,1\}^{V(A)})$. Here, $G_A=\prod_{\{ij\}\in A}G_{ij}$ is the boolean Fourier character and $\chi^q_B(x) = \prod_{i\in B}\frac{x_i-q}{\sqrt{q(1-q)}}$ is the $q$-biased character.Similarly, define $W'_{A,B}\subseteq L^2(\Psi')$ as the space of all functions of the form
\[(y,H)\to H_A\cdot \chi^{pq}_B(y)\cdot r(y_{V(A)}).\]
The following is a direct result of the fact that the Fourier characters form an orthogonal basis, and the fact that $\Psi$ and $\Psi'$ are product distributions.

\begin{fact}\label{fact:sub-orth}
    The collection of subspaces $W_{A,B}$ (resp. $W'_{A,B}$) form an orthogonal decomposition of $L^2(\Psi)$ (resp. $L^2(\Psi')$), that is, they are an orthogonal collection of subspaces that span the entire space.
\end{fact}

We will require two results about these subspaces. 

\begin{claim}\label{claim:contr}
    For any valid pair $(A,B)$, $T^*$ is a $p^{\frac{|V(A)|+|B|}{2}}$-contractive mapping from $W'_{A,B}$ to $W_{A,B}$. That is, for any function $w\in W'_{A,B}$, $T^*w\in W_{A,B}$, and $\|T^*w\|_{2,\Psi}\leq p^{\frac{|V(A)|+|B|}{2}}\cdot \|w\|_{2,\Psi'}$.
\end{claim}

\begin{claim}\label{claim:anti-low}
    Any $S_n$-symmetric function in $\Span\{W_{A,B}:|V(A)|+|B|\leq 4d,(A,B)\}$ is $((2c)^{d^2},1/2)$-anticoncentrated for some absolute constant $c$.
\end{claim}

We prove both of the above results in \cref{subsec:hyp-aux}. For now we will use them to obtain \cref{lemma:planted-clique-hyp}. Let $f\in L^2(\pla')$ be the given function that satisfies $\|Tf\|_{2,\pla}\geq p^d\cdot \|f\|_{2,\pla'}$. We will define $g\in L^2(\Psi')$ by $g(y,H) = f(H\wedge \clique(y))$. By \cref{obs:equiv}, we have $\|T^*g\|_{2,\Psi}\geq p^d\cdot \|g\|_{2,\Psi'}$.
We will use \cref{fact:sub-orth} to write $g=g_+ + g_-$, where
\[g_+\in \Span\{W_{A,B}:|V(A)|+|B|\leq 4d\},\quad g_-\in \Span\{W_{A,B}:|V(A)|+|B|> 4d\}.\]
By \cref{claim:contr}, we get
\begin{equation}\label{eqn:killed-g}
    \|T^*g_-\|_{2,\Psi}\leq p^{2d}\cdot \|g_-\|_{2,\Psi'}\leq p^{2d}\cdot \|g\|_{2,\Psi}\leq p^d\cdot \|T^*g\|_{2,\Psi}.
\end{equation}
Let us calculate the probability that $T^*g$ is at least $\|T^*g\|_{2,\Psi}/10$. We will lower bound it by the probability that both (a) $|T^*g_{+}|\geq \|T^*g\|_{2,\Psi}/4$ and (b) $|T^*g_{-}|< \|T^*g\|_{2,\Psi}/10$.

    \begin{enumerate}
        \item[(a)] \cref{eqn:killed-g} implies $\|T^*g_{+}\|_{2,\Phi}^2 = \|T^*g\|_{2,\Phi}^2-\|T^*g_{-}\|_{2,\Phi}^2\geq (1-p^{2d})\cdot \|T^*g\|_{2,\Phi}^2\geq \|T^*g\|_{2,\Phi}^2/2$. We conclude by \cref{claim:anti-low} that $|T^*g_{+}(x,G)|\geq \|T^*g(\Psi)\|_{2,\Psi}/4$ with probability at least $c^{d^2}$.
        \item[(b)] By Markov's inequality,
        \begin{align*}
        \Pr_{\Psi}[|Tg_{-}|>\|T^*g\|_{2,\Psi}/10]&=\Pr_{\Psi}[|Tg_{-}|^2>\|T^*g\|_{2,\Psi}^2/100]\\
        &\leq \frac{100\|T^*g_{-}\|_{2,\Psi}^2}{\|T^*g\|_{2,\Psi}^2}\\
        &\leq  100 p^{2d}.\tag{\cref{eqn:killed-g}}
    \end{align*}
    \end{enumerate}
    As a result, $T^*g$ is at least $\|T^*g\|_{2,\Psi}/10$ with probability $(2c)^{-d^2} - 100p^{2d} \geq c^{-d^2}$, where we used that $d=o(\log p^{-1})$. That is, $T^*g$ is $(c^{-d^2},0.1)$-anticoncentrated. By \cref{obs:equiv}, this is equivalent to the statement that $Tf$ is $(c^{-d^2}, 0.1)$-anticoncentrated, completing the proof.

\subsection{Proof of Auxiliary Claims}\label{subsec:hyp-aux}

It remains to prove \cref{claim:contr,claim:anti-low}.
\begin{proof}[Proof of \cref{claim:contr}]
    Let $w\in W'_{A,B}$, that is, $w(y,H)=H_A\cdot \chi_B^{pq}(y)\cdot r(y_{V(A)})$ for an arbitrary function $r$. We need to prove that $T^*w\in W_{A,B}$, and that $\|T^*w\|_{2,\Psi}\leq p^{\frac{|V(A)|+|B|}{2}}\cdot \|w\|_{2,\Psi'}$. Let us write an expression for $T^*w$. Let $z\sim \ber(p)^{\otimes n}$. Since $B$ is disjoint from $V(A)$, we can write
    \begin{align*}
      T^*w(x,G) &= \E_{z}[\E_{H\mid G,z}[H_A]\cdot \chi_B^{pq}(x\wedge z)\cdot r((x\wedge z)_{V(A)})]\\
      &= \left(\prod_{i\in B}\E_{z}[\chi_{\{i\}}^{pq}(x\wedge z)]\right)\cdot \E_{z}[\E_{H\mid G,z}[H_A]\cdot r((x\wedge z)_{V(A)})],
    \end{align*}
    We will consider each term separately. First, let $i\in B$. We have
    \begin{align*}
      \E_{z}[\chi_{\{i\}}^{pq}(x\wedge z)] &= \E_{z_i\sim \ber(p)}\left[\frac{x_iz_i - pq}{\sqrt{pq(1-pq)}}\right]\\
      &= \frac{p\cdot (x_i - p)}{\sqrt{pq(1-pq)}} = \sqrt{\frac{p\cdot (1-p)}{1-pq}}\cdot \chi_{\{i\}}^{q}(x).
    \end{align*}
    
    Next, we will simplify the second term in our expression for $T^*w(x,G)$. Recall that conditioned on values for $G$ and $z$, $H_{\{i,j\}}$ is set to $G_{\{i,j\}}$ if $z_i=z_j=1$, and an independent Rademacher otherwise. As a result we have $\E_{H\mid G,z}[H_A] = \mathbf{1}_{\{z_{V(A)}=1\}}\cdot G_A$, so
    \[\E_{z}[\E_{H\mid G,z}[H_A]\cdot r((x\wedge z)_{V(A)})]=p^{|V(A)|}\cdot G_A\cdot r(x_{V(A)}).\]
    Combining these two expressions, we can write
    \[T^*w(x, G)=p^{|V(A)|+|B|/2}\cdot \sqrt{\frac{1-p}{1-pq}}^{|B|}\cdot G_A\cdot \chi_B^{q}(x)\cdot r(x_{V(A)}).\]
    This proves that $T^*w\in W_{A,B}$. It remains to show that $T^*$ contracts the norm of $w$. 
    \begin{align*}
        \|T^*w\|_2^2&\leq p^{2|V(A)|+|B|}\cdot \E_{(x,G)\sim \Psi}\left[\left({G}_A\cdot \chi_B^q(x)\cdot r(x_{V(A)})\right)^2\right]\tag{$\sqrt{\frac{1-p}{1-pq}}\leq 1$}\\
        &=p^{2|V(A)|+|B|}\cdot \E_{x\sim \ber(q)^{\otimes n}}\left[ \chi_B^q(x)^2\cdot r(x_{V(A)})^2\right]\tag{$G_A\in \{\pm 1\}$}\\
        &= p^{2|V(A)|+|B|}\cdot \E_{x\sim \ber(q)^{\otimes n}}\left[r(x_{V(A)})^2\right].\tag{$B\cap V(A) = \emptyset$, $\E[\chi_B^q(x)^2]=1$}
    \end{align*}
    We will bound this further by noting that $\E_{x\sim \ber(q)^{\otimes n}}\left[r(x_{V(A)})^2\right] = \E_{X\sim \ber(q)^{\otimes |V(A)|}}[R(X)]$ for some nonnegative function $R$. We will bound
    \[\E_{X\sim \ber(q)^{\otimes |V(A)|}}[R(X)]\leq \left\|\frac{d\ber(q)^{\otimes |V(A)|}}{d\ber(pq)^{\otimes |V(A)|}}\right\|_\infty \E_{X\sim \ber(pq)^{\otimes |V(A)|}}[R(X)].\]
    Calculating $\left\|\frac{d\ber(q)^{\otimes |V(A)|}}{d\ber(pq)^{\otimes |V(A)|}}\right\|_\infty = \left\|\frac{d\ber(q)}{d\ber(pq)}\right\|_\infty^{|V(A)|}=p^{-|V(A)|}$, we get the bound
    \[\|T^*w\|_2^2\leq p^{|V(A)|+|B|}\cdot \E_{y\sim \ber(pq)^{\otimes n}}[r(y_{V(A)})^2].\]
    On the other hand, we can compute
    \[\|w\|_2^2 = \E_{(y,H)\sim \Psi'}[(H_A\cdot \chi_B^{pq}(y)\cdot r(y_{V(A)}))^2] = \E_{y\sim \ber(pq)^{\otimes n}}[r(y_{V(A)})^2]\;.\]
    We conclude that $\|T^*w\|_2\leq p^{\frac{|V(A)| + |B|}{2}}\cdot \|w\|_2$, completing the proof.
\end{proof}

\begin{proof}[Proof of \cref{claim:anti-low}]
    We must show that any $S_n$-symmetric function $h\in \Span\{W_{A,B}:|V(A)|+|B|\leq 4d\}$ is $((2c)^{d^2},1/2)$-anticoncentrated for some absolute constant $c$. We will prove an upper bound on $\frac{\|h\|_{8,\Psi}}{\|h\|_{4,\Psi}}$ and then use the Paley-Zygmund inequality (\cref{lem:paley-zygmund}) to obtain anticoncentration. Let us write 
    \[\|h\|_{8,\Psi}^8 = \E_\Psi[h^8] = \E_{x}[\E_{G}[h^8(x,G)]].\]
    \begin{observation}
        For any fixed $x$, the map $G\to h(x,G)$ is a degree $\leq \binom{|V(A)|}{2}\leq 8d^2$ function over the boolean hypercube $\{\pm 1\}^{\binom{[n]}{2}}$. As a consequence, $G\to h^2(x,G)$ is a degree $\leq 16d^2$ function.
    \end{observation}
    Consequently for any $x$ we can use Bonami's lemma (\cref{lem:bonami}) twice to bound the inner expectation by
    \[\E_{G\sim G(n,1/2)}[h^8(x,G)]\leq 9^{8d^2}\E_{G\sim G(n,1/2)}[h^4(x,G)]^2\leq 9^{16d^2}\E_{G\sim G(n,1/2)}[h^2(x,G)]^4.\]
    Now consider the function $s(x) = \E_{G}[h^2(x,G)]$ under the $q$-biased hypercube, that is, under $x\sim \ber(q)^{\otimes n}$. Since $h$ is symmetric under permutations $\pi\in S_n$, $s$ is also $S_n$-symmetric and has degree $\leq 8d = o(np)$, so we can apply \cref{lem:p-biased} to obtain $\E[s^4(x)\leq 8^{32d}\cdot \E[s^2(x)]^2$. Combining these two bounds,
    \begin{align*}
      \|h\|_{8,\Psi}^8&\leq 9^{16d^2}\cdot \E_x[s^4(x)]\\
      &\leq 9^{16d^2} 8^{32d}\cdot\E_x[s^2(x)]^2\\
      &= 9^{16d^2} 8^{32d}\cdot\E_x[\E_{G}[h^2(x,G)]^2]^2\\
      &\leq 9^{16d^2} 8^{32d}\cdot\E_{x}[\E_{G}[h^4(x,G)]]^2 = 9^{16d^2} 8^{32d}\cdot\|h\|_{4,\Psi}^8.\tag{Jensen's inequality}
    \end{align*}
    We will apply the Paley-Zygmund inequality (\cref{lem:paley-zygmund}) on $h^2$ to obtain anticoncentration of $h$.
    \begin{align*}
        \Pr_\Psi[|h|\geq \|h\|_{2,\Psi}/2]&\geq \Pr_\Psi[|h|\geq \|h\|_{4,\Psi}/2]\tag{Jensen's inequality}\\
        &= \Pr_\Psi[h^2\geq \|h^2\|_{2,\Psi}/4]\\
        &\geq \frac{\E_\Psi[h^4]^2}{4\E_\Psi[h^8]}\tag{\cref{lem:paley-zygmund}}\\
        &\geq \frac{1}{4\cdot 9^{16d^2} 8^{32d}}\geq (2c)^{d^2}.\tag{for $2c=(4\cdot 9^{16}\cdot 8^{32})^{-1}$}
    \end{align*}
\end{proof}

\section{Proof of Main Theorems for Planted Clique}\label{sec:application}

In this section we will prove our results on Planted Clique. Having already proved all the pieces in previous sections, this section will not introduce any new ideas; it will mostly involve invoking these pieces with an appropriate choice of parameters.

\subsection{Central Theorem}
We will begin by stating and proving a formal version of \cref{thm:perturbation}, which will imply all our other results.

\begin{theorem}\label{thm:perturbation-full}
    Assume the Planted Clique Hypothesis (\cref{conj:planted}). Let $0<\alpha<\beta$, and let $d\in \N$. Consider $\pla=G(n,1/2,n^{1/2-\alpha})$ and $\nul=G(n,1/2)$. Let $\pla'$ be another distribution over $\Omega$. Define $\pla^*=M^{\sym}(\pla')$, where $M$ is the permuted vertex-resampling Markov chain with noise rate $1-p$ for $p=n^{\alpha-\beta}$. For every randomized polynomial time test $A$,
    \[\algadv^{(\pla^*,\nul)}(A)\leq \var_{\nul}(\Pi_{\F_{\leq d}} f)^{1/2}\cdot\adv^{(\pla^*,\nul)}[\F_{\leq d}] + O(n^{-(\beta-\alpha)d'})\cdot \left\|\frac{d\pla'}{d\pla}\right\|_\infty,\]
    where $f(G) = \mathbb{E}_{A}[A(G)]$, $\var_{\nul}(\Pi_V f) = \|\Pi_V f - \E_{\nul}[\Pi_V f]\|_{2,\nul}^2$, and $d'$ is the minimum number of vertices in any simple graph with $d+1$ edges. In particular, we have $\var_{\nul}(\Pi_{\F_{\leq d}} f)^{1/2}\leq \|f\|_{2,\nul}\leq 1$ and $d'\geq \sqrt{d}$, implying the simplified statement
        \[\algadv^{(\pla^*,\nul)}(A)\leq \adv^{(\pla^*,\nul)}[\F_{\leq d}] + O(n^{-(\beta-\alpha)\sqrt{d}})\cdot \left\|\frac{d\pla'}{d\pla}\right\|_\infty.\]
\end{theorem}
\begin{proof}[Proof of \cref{thm:perturbation-full}]
    We will argue the contrapositive. Suppose there is a polynomial time computable test $A$ such that for infinitely many values of $n$,
    \begin{equation}\label{eqn:a-adv-application}
      \algadv^{(\pla^*,\nul)}(A)\geq \| \Pi_{\F_{\leq d}} f\|_{2,\nul}\cdot\adv^{(\pla^*,\nul)}[\F_{\leq d}]+\omega(n^{-(\beta-\alpha)d'})\cdot \left\|\frac{d\pla'}{d\pla}\right\|_\infty  
    \end{equation}
    In light of \cref{thm:hirahara}, we would like to conclude that there is a polynomial time computable test $B$ such that for infinitely many values of $n$, 
    \[\algadv^{(\pla,\nul)}(B)=\Omega(1),\]
    which would imply that the Planted Clique Hypothesis is false.
    To do this, we will apply \cref{thm:general-1} with an appropriate choice of parameters. We will set
    \[\eps=p^{d'}= n^{(\alpha-\beta)d'},\quad q=n^{2d},\quad \gamma = c^{((4d/(\beta-\alpha))^2}\]
    for $c$ as in \cref{lemma:planted-clique-hyp}. We will set $V=\F_{\leq d}$ to be the space of degree $\leq d$ polynomials. Let $T$ be the noise operator associated with $M^{\sym}$.
    
    \begin{fact}\label{fact:inclusion}
        For any $0<p<1$ and $d$, $\F_{\leq d}$ contains the top eigenspace $\lambda_{> \eps}(T)$ for $\eps = p^{d'}$, where $d'$ is the minimum number of vertices an any simple graph with $d+1$ edges.
    \end{fact}
    
    Let us verify all requirements of \cref{thm:general-1}.

    \begin{itemize}
        \item One can sample from $\nul$, and sample $H\sim M^{\sym}(G)$ in time $O(n^2) = O(q)$.
        \item $V=\F_{\leq d}$ is $\binom{n}{2}^d\leq q$-tractable.
        \item By \cref{lemma:planted-clique-hyp}, if $\|Tf\|_{2,\pla}\geq \eps/q\cdot \|f\|_{2,M^{\sym}(\pla)}\geq n^{-4d}\cdot \|f\|_{2,M^{\sym}(\pla)}$, we have that $Tf(\pla)$ is $(\gamma,0.1)$-anticoncentrated by definition of $\gamma$.
    \end{itemize}
    By \cref{eqn:a-adv-application}, we have $\algadv^{(\pla^*,\nul)}(A)\geq \adv^{(\pla^*,\nul)}(\F_{\leq d})+\omega(\delta)\cdot \left\|\frac{d\pla'}{d\pla}\right\|_\infty$, where $\delta=1600\gamma^{-1}\eps = O(p^{d'}) = O(n^{-(\beta-\alpha)d'})$. Therefore we can apply \cref{thm:general-1} to deduce the existence of an efficiently computable test $B$ achieving $\algadv^{(\pla,\nul)}(B)\geq \gamma=\Omega(1)$ for infinitely many $n$, as required.
\end{proof}

\subsection{Implications for Vanilla Planted Clique}

\cref{thm:main-1} follows directly from \cref{thm:perturbation-full} -- it is exactly equivalent to the simplified statement in the special case that $\pla' = \pla$. Next we will deduce \cref{cor:ldlr-value}.

\begin{proof}[Proof of \cref{cor:ldlr-value}]
  Let us set $d=1$. We have that $d'$, the minimum number of vertices in any simple graph with $d+1$ edges, equals $3$. The following claim is a simple calculation which we defer to \cref{sec:a3}.
    \begin{claim}\label{claim:ldlr-exact}
    Let $\pla=G(n,1/2,k)$ and $\nul=G(n,1/2)$ for $k=n^{1/2-\beta}$ for some $\beta>0$. We have $\adv^{(\pla,\nul)}[\F_{\leq d}] = (1+o_n(1))\cdot k^2/(\sqrt{2}n)$ for any constant $d\geq 1$.
    \end{claim}

    Set $\alpha = \beta/10$ for some arbitrary $0<\beta < 1$. We will apply the strong bound in \cref{thm:perturbation-full} for $\pla'=\pla$. We have $\pla^* = G(n, 1/2, k)$ for $k=n^{1/2-\beta}$. Let $A$ be a randomized polynomial time test. Since $\pla^*$ and $\nul$ are both invariant under permuting the $n$ vertices of the graph, we can assume at no cost to $A$'s advantage that $A$ is also invariant under permuting the vertices of the input graph. Applying \cref{thm:perturbation-full}, we conclude that for any randomized polynomial time test $A$,
    \[\algadv^{(\pla^*, \nul)}(A)\leq \var_{\nul}(\Pi_{\F_{\leq 1}} f)^{1/2}\cdot \frac{k^2}{\sqrt{2}n} + O(n^{-0.9\cdot \beta\cdot 3}) = \var_{\nul}(\Pi_{\F_{\leq 1}} f)^{1/2}\cdot \frac{k^2}{\sqrt{2}n} + o\left(\frac{k^2}{\sqrt{n}}\right).\]
    To complete the proof, we will argue that $\var_{\nul}(\Pi_{\F_{\leq 1}} f)^{1/2} \leq \sqrt{\frac{2}{\pi}} + o_n(1)$.

    \begin{claim}\label{claim:linear-symmetric}
        Let $p:\{\pm 1\}^T\to [-1,1]$ be a bounded function over the $m$-dimensional boolean hypercube with Fourier decomposition $p=\sum_{S\subseteq T}\hat{p}_S\cdot \chi_S$. If its degree-1 part $p_{=1}=\sum_{i\in T}\hat{p}_{\{i\}}\chi_{\{i\}}$ is symmetric under permuting the $|T|$ coordinates, then $\|p_{=1}\|_{2,\nul}\leq \sqrt{2/\pi} + o_{|T|}(1) $ where $\nul = \unif(\{\pm 1\}^T)$.
    \end{claim}
    \begin{proof}
        By symmetry, we have $\hat{p}_{\{i\}} = \hat{p}_{\{j\}}$ for all $i,j\in T$. So, 
        \begin{align*}
            |\hat{p}_{\{j\}}| &= |T|^{-1}\cdot \left|\sum_{i\in T}\hat{p}_{\{i\}}\right|\\
            &= |T|^{-1}\cdot \left|\mathbb{E}_{x\sim \{\pm 1\}^T}[p(x)\cdot \sum_{i\in T}x_i]\right|\\
            &\leq |T|^{-1}\cdot \E_{x\sim \{\pm 1\}^T} \left|\sum_{i\in T}x_i\right|\tag{$p$ has range $[-1,1]$}
        \end{align*}

        By the Central Limit Theorem, we have that the expected value of $\left|\sum_{i\in T}x_i\right|$ is at most $(1+o_{|T|}(1))\cdot \sqrt{|T|}\cdot \E_{g\sim N(0,1)}[|g|] = (1+o_{|T|}(1)){}\sqrt{|T|}\cdot \sqrt{\frac{2}{\pi}}$. Therefore we have 
        \[\|p_{=1}\|_{2,\nul} = \sqrt{\sum_{i\in T}\hat{p}_{\{i\}}^2} \leq (1+o_{|T|}(1))\cdot \sqrt{\frac{2}{\pi}}.\]
    \end{proof}
    We will apply \cref{claim:linear-symmetric} to $p=f-\mathbb{E}_\nul[f]$ for $T = \binom{[n]}{2}$. Since $f$ is symmetric under permuting the vertices of the input graph, its degree-1 Fourier coefficients are indeed symmetric under permuting the $\binom{n}{2}$ potential edges -- for any two edges $e,e'\in \binom{[n]}{2}$, there is a permutation on vertices that sends $e$ to $e'$. Therefore, we get $\var_{\nul}(\Pi_{\F_{\leq 1}}f)^{1/2} = \|p_{=1}\|_{2,\nul}\leq \sqrt{2/\pi} + o_n(1)$. We have now proved the desired bound
    \[\algadv^{(\pla^*, \nul)}(A)\leq (1+o_n(1))\cdot \frac{k^2}{\sqrt{\pi}n}.\]
\end{proof}

\subsection{Hard-core Lemma style results}

Finally, we prove our hard-core lemma style results, beginning with \cref{thm:intro-moment-match}, which states that one can find a perturbation $\pla'$ of $\pla$ that is arbitrarily hard for low-degree polynomials. We will do this by invoking \cref{thm:general-2}.

\begin{proof}[Proof of \cref{thm:intro-moment-match}]
    We will simply apply \cref{thm:general-2} to $\pla = G(n,1/2,k)$ for $k=n^{1/2-\alpha}$ to $V=\F_{\leq d}$. Set $q(n) = n^{O(d)}$ and $c=\sqrt{3}^d=O(1)$. Let us check that the conditions of the \cref{thm:general-2} hold.
    \begin{itemize}
        \item One can sample from $\pla$ in time $O(n^2)\leq q(n)$.
        \item $V$ is indeed $\binom{n}{2}^d \leq q$-tractable.
        \item For hypercontractivity, we will use Bonami's lemma (\cref{lem:bonami}) on $f^2\in \F_{\leq 2d}$, which states that $\|f^2\|_{4,\nul}\leq 3^{d}\cdot\|f^2\|_{2,\nul}$ for $f\in \F_{\leq d}$. Let us bound
        \[ \E_\nul[f^8] = \|f^2\|_{4,\nul}^4\leq 3^{4d}\cdot \E_\nul[f^4]^2.\]
        So, $\|f\|_{8,\nul}\leq \sqrt{3}^d\cdot \|f\|_{4,\nul}$.
        \item We have by \cref{claim:ldlr-exact} that $\adv^{(\pla,\nul)}[\F_{\leq d}] = o_n(1) = o(c)$.
    \end{itemize}
    The statement of \cref{thm:general-2} then directly implies \cref{thm:intro-moment-match}.
\end{proof}

Finally, we show how to combine \cref{thm:perturbation-full,thm:intro-moment-match} to get \cref{thm:main-clique}.

\begin{proof}[Proof of \cref{thm:main-clique}]
Let $\alpha>0$, and set $\beta=2\alpha$. We will set $d = 10D^2/\alpha$ and $\delta = n^{-2D}$. Let us apply \cref{thm:intro-moment-match} to obtain $\pla'$ satisfying $\left\|\frac{d\pla'}{d\pla}\right\|_\infty = 1+o_n(1)$ and $\R^{(\pla',\nul)}(A)\leq \delta$. Applying \cref{thm:perturbation-full} to $\pla'$, we obtain $\pla^*$ so that for any efficient test $A$,
\[\algadv^{(\pla^*,\nul)}(A)\leq \adv^{(\pla^*,\nul)}[\F_{\leq d}] + O(n^{-2D})\leq \adv^{(\pla',\nul)}[\F_{\leq d}] + O(n^{-2D}) \leq \delta + O(n^{-2D}).\]
By definition of $\delta$, this is $O(n^{-2D}) = o(n^{-D})$ as required. Note that $\pla^*$ does not contain a large clique with probability $1$. Instead, we can show that $T(\pla)$ contains a clique of size, say $n^{1/2-10\beta}$ except with exponentially small probability $\exp(-O(n))$. Since $\|\tfrac{d\pla^*}{M^{\sym}(\pla)}\|_\infty\leq \|\tfrac{d\pla'}{\pla}\|_\infty=1+o_n(1)$, this implies that $\pla^*$ also contains a clique of size $n^{1/2-10\beta}$ with probability $1-\exp(-O(n))$. So the sampler can keep track of the sampled clique and use rejection sampling to ensure that the resulting graph has a large clique with probability $1$.

This results in a $\exp(-O(n))$ perturbation to $\pla^*$ in TV distance, which can only change the value of $\algadv^{(\pla^*,\nul)}(A)$ by $\exp(-O(n))$, and the theorem follows.
\end{proof}


\printbibliography

@misc{moitra2023preciseerrorratescomputationally,
	title        = {Precise Error Rates for Computationally Efficient Testing},
	author       = {Ankur Moitra and Alexander S. Wein},
	year         = 2023,
	url          = {https://arxiv.org/abs/2311.00289},
	eprint       = {2311.00289},
	archiveprefix = {arXiv},
	primaryclass = {math.ST}
}

@article{hazan2011hard,
	title        = {How hard is it to approximate the best Nash equilibrium?},
	author       = {Hazan, Elad and Krauthgamer, Robert},
	year         = 2011,
	journal      = {SIAM Journal on Computing},
	publisher    = {SIAM},
	volume       = 40,
	number       = 1,
	pages        = {79--91}
}

@article{manurangsi2020strongish,
	title        = {The strongish planted clique hypothesis and its consequences},
	author       = {Manurangsi, Pasin and Rubinstein, Aviad and Schramm, Tselil},
	year         = 2020,
	journal      = {arXiv preprint arXiv:2011.05555}
}

@inproceedings{alon2007testing,
	title        = {Testing k-wise and almost k-wise independence},
	author       = {Alon, Noga and Andoni, Alexandr and Kaufman, Tali and Matulef, Kevin and Rubinfeld, Ronitt and Xie, Ning},
	year         = 2007,
	booktitle    = {Proceedings of the thirty-ninth annual ACM symposium on Theory of computing},
	pages        = {496--505}
}

@inproceedings{abram2023cryptography,
	title        = {Cryptography from planted graphs: security with logarithmic-size messages},
	author       = {Abram, Damiano and Beimel, Amos and Ishai, Yuval and Kushilevitz, Eyal and Narayanan, Varun},
	year         = 2023,
	booktitle    = {Theory of Cryptography Conference},
	pages        = {286--315},
	organization = {Springer}
}

@inproceedings{bogdanov2024low,
	title        = {Low-Degree Security of the Planted Random Subgraph Problem},
	author       = {Bogdanov, Andrej and Jones, Chris and Rosen, Alon and Zadik, Ilias},
	year         = 2024,
	booktitle    = {Theory of Cryptography Conference},
	pages        = {255--275},
	organization = {Springer}
}

@inproceedings{applebaum2010public,
	title        = {Public-key cryptography from different assumptions},
	author       = {Applebaum, Benny and Barak, Boaz and Wigderson, Avi},
	year         = 2010,
	booktitle    = {Proceedings of the forty-second ACM symposium on Theory of computing},
	pages        = {171--180}
}

@article{juels2000hiding,
	title        = {Hiding cliques for cryptographic security},
	author       = {Juels, Ari and Peinado, Marcus},
	year         = 2000,
	journal      = {Designs, Codes and Cryptography},
	publisher    = {Springer},
	volume       = 20,
	number       = 3,
	pages        = {269--280}
}

@inproceedings{elrazik2022pseudorandom,
	title        = {Pseudorandom self-reductions for NP-complete problems},
	author       = {Elrazik, Reyad Abed and Robere, Robert and Schuster, Assaf and Yehuda, Gal},
	year         = 2022,
	booktitle    = {13th Innovations in Theoretical Computer Science Conference (ITCS 2022)},
	pages        = {65--1},
	organization = {Schloss Dagstuhl--Leibniz-Zentrum f{\"u}r Informatik}
}

@article{kuvcera1995expected,
	title        = {Expected complexity of graph partitioning problems},
	author       = {Ku{\v{c}}era, Lud{\v{e}}k},
	year         = 1995,
	journal      = {Discrete Applied Mathematics},
	publisher    = {Elsevier},
	volume       = 57,
	number       = {2-3},
	pages        = {193--212}
}

@article{jerrum1992large,
	title        = {Large cliques elude the Metropolis process},
	author       = {Jerrum, Mark},
	year         = 1992,
	journal      = {Random Structures \& Algorithms},
	publisher    = {Wiley Online Library},
	volume       = 3,
	number       = 4,
	pages        = {347--359}
}

@article{krivelevich2002approximating,
	title        = {Approximating the independence number and the chromatic number in expected polynomial time},
	author       = {Krivelevich, Michael and Vu, Van H},
	year         = 2002,
	journal      = {Journal of combinatorial optimization},
	publisher    = {Springer},
	volume       = 6,
	pages        = {143--155}
}

@inproceedings{hajek2015computational,
	title        = {Computational lower bounds for community detection on random graphs},
	author       = {Hajek, Bruce and Wu, Yihong and Xu, Jiaming},
	year         = 2015,
	booktitle    = {Conference on Learning Theory},
	pages        = {899--928},
	organization = {PMLR}
}

@inproceedings{bresler2023detection,
	title        = {Detection-recovery and detection-refutation gaps via reductions from planted clique},
	author       = {Bresler, Guy and Jiang, Tianze},
	year         = 2023,
	booktitle    = {The Thirty Sixth Annual Conference on Learning Theory},
	pages        = {5850--5889},
	organization = {PMLR}
}

@inproceedings{HopkinsKPRSS17,
  author    = {Samuel B. Hopkins and
               Pravesh K. Kothari and
               Aaron Potechin and
               Prasad Raghavendra and
               Tselil Schramm and
               David Steurer},
  title     = {The Power of Sum-of-Squares for Detecting Hidden Structures},
  booktitle = {58th {IEEE} Annual Symposium on Foundations of Computer Science, {FOCS}
               2017, Berkeley, CA, USA, October 15-17, 2017},
  pages     = {720--731},
  year      = {2017},
 }

@inproceedings{barak2009uniform,
	title        = {The uniform hardcore lemma via approximate bregman projections},
	author       = {Barak, Boaz and Hardt, Moritz and Kale, Satyen},
	year         = 2009,
	booktitle    = {Proceedings of the twentieth annual ACM-SIAM symposium on Discrete algorithms},
	pages        = {1193--1200},
	organization = {SIAM}
}

@inproceedings{holenstein2005key,
	title        = {Key agreement from weak bit agreement},
	author       = {Holenstein, Thomas},
	year         = 2005,
	booktitle    = {Proceedings of the thirty-seventh annual ACM symposium on Theory of computing},
	pages        = {664--673}
}

@inproceedings{klivans1999boosting,
	title        = {Boosting and hard-core sets},
	author       = {Klivans, Adam R and Servedio, Rocco A},
	year         = 1999,
	booktitle    = {40th Annual Symposium on Foundations of Computer Science (Cat. No. 99CB37039)},
	pages        = {624--633},
	organization = {IEEE}
}

@inproceedings{yao1982theory,
	title        = {Theory and application of trapdoor functions},
	author       = {Yao, Andrew C},
	year         = 1982,
	booktitle    = {23rd Annual Symposium on Foundations of Computer Science (SFCS 1982)},
	pages        = {80--91},
	organization = {IEEE}
}

@inproceedings{I95,
	title        = {Hard-core distributions for somewhat hard problems},
	author       = {Impagliazzo, R.},
	year         = 1995,
	booktitle    = {Proceedings of IEEE 36th Annual Foundations of Computer Science},
	pages        = {538--545},
	doi          = {10.1109/SFCS.1995.492584}
}

@inproceedings{hirahara2024planted,
	title        = {Planted Clique Conjectures Are Equivalent},
	author       = {Hirahara, Shuichi and Shimizu, Nobutaka},
	year         = 2024,
	booktitle    = {Proceedings of the 56th Annual ACM Symposium on Theory of Computing},
	pages        = {358--366}
}

@article{garrigos2023handbook,
	title        = {Handbook of convergence theorems for (stochastic) gradient methods},
	author       = {Garrigos, Guillaume and Gower, Robert M},
	year         = 2023,
	journal      = {arXiv preprint arXiv:2301.11235}
}

@inproceedings{bonami1970etude,
	title        = {{\'E}tude des coefficients de Fourier des fonctions de $ L^{p} (G) $},
	author       = {Bonami, Aline},
	year         = 1970,
	booktitle    = {Annales de l'institut Fourier},
	volume       = 20,
	number       = 2,
	pages        = {335--402}
}

@inproceedings{hirahara2023hardness,
	title        = {Hardness self-amplification: Simplified, optimized, and unified},
	author       = {Hirahara, Shuichi and Shimizu, Nobutaka},
	year         = 2023,
	booktitle    = {Proceedings of the 55th Annual ACM Symposium on Theory of Computing},
	pages        = {70--83}
}

@article{keevash2021global,
	title        = {Global hypercontractivity and its applications},
	author       = {Keevash, Peter and Lifshitz, Noam and Long, Eoin and Minzer, Dor},
	year         = 2021,
	journal      = {arXiv preprint arXiv:2103.04604}
}

@inproceedings{CHKRT19,
	title        = {Algorithms for heavy-tailed statistics: Regression, covariance estimation, and beyond},
	author       = {Cherapanamjeri, Yeshwanth and Hopkins, Samuel B and Kathuria, Tarun and Raghavendra, Prasad and Tripuraneni, Nilesh},
	year         = 2020,
	booktitle    = {Proceedings of the 52nd Annual ACM SIGACT Symposium on Theory of Computing},
	pages        = {601--609}
}

@article{BKW19,
	title        = {Computational hardness of certifying bounds on constrained PCA problems},
	author       = {Bandeira, Afonso S and Kunisky, Dmitriy and Wein, Alexander S},
	year         = 2019,
	journal      = {arXiv preprint arXiv:1902.07324}
}

@article{BB19,
	title        = {Average-case lower bounds for learning sparse mixtures, robust estimation and semirandom adversaries},
	author       = {Brennan, Matthew and Bresler, Guy},
	year         = 2019,
	journal      = {arXiv preprint arXiv:1908.06130}
}

@article{BB20,
	title        = {{Reducibility and Statistical-Computational Gaps from Secret Leakage}},
	author       = {Brennan, Matthew and Bresler, Guy},
	year         = 2020,
	journal      = {arXiv preprint arXiv:2005.08099}
}

@article{BHKKMP19,
	title        = {A nearly tight sum-of-squares lower bound for the planted clique problem},
	author       = {Barak, Boaz and Hopkins, Samuel and Kelner, Jonathan and Kothari, Pravesh K and Moitra, Ankur and Potechin, Aaron},
	year         = 2019,
	journal      = {SIAM Journal on Computing},
	publisher    = {SIAM},
	volume       = 48,
	number       = 2,
	pages        = {687--735}
}

@misc{MRX19,
	title        = {Lifting Sum-of-Squares Lower Bounds: Degree-$2$ to Degree-$4$},
	author       = {Sidhanth Mohanty and Prasad Raghavendra and Jeff Xu},
	year         = 2019,
	eprint       = {1911.01411},
	archiveprefix = {arXiv},
	primaryclass = {cs.CC}
}

@inproceedings{DBLP:conf/colt/BrennanBH18,
	title        = {{Reducibility and Computational Lower Bounds for Problems with Planted Sparse Structure}},
	author       = {Matthew Brennan and Guy Bresler and Wasim Huleihel},
	year         = 2018,
	booktitle    = {Conference On Learning Theory, {COLT} 2018, Stockholm, Sweden, 6-9 July 2018.},
	pages        = {48--166},
	url          = {http://proceedings.mlr.press/v75/brennan18a.html},
	bibsource    = {dblp computer science bibliography, https://dblp.org}
}

@inproceedings{berthet2013complexity,
	title        = {Complexity theoretic lower bounds for sparse principal component detection},
	author       = {Berthet, Quentin and Rigollet, Philippe},
	year         = 2013,
	booktitle    = {Conference on Learning Theory},
	pages        = {1046--1066}
}

@inproceedings{KWB19,
	title        = {Notes on computational hardness of hypothesis testing: Predictions using the low-degree likelihood ratio},
	author       = {Kunisky, Dmitriy and Wein, Alexander S and Bandeira, Afonso S},
	year         = 2019,
	booktitle    = {ISAAC Congress (International Society for Analysis, its Applications and Computation)},
	pages        = {1--50},
	organization = {Springer}
}

@inproceedings{hopkins2017efficient,
	title        = {Efficient Bayesian estimation from few samples: community detection and related problems},
	author       = {Hopkins, Samuel B and Steurer, David},
	year         = 2017,
	booktitle    = {Foundations of Computer Science (FOCS), 2017 IEEE 58th Annual Symposium on},
	pages        = {379--390},
	organization = {IEEE}
}

@inproceedings{Hopkins18thesis,
	title        = {Statistical Inference and the Sum of Squares Method.},
	author       = {Samuel Hopkins},
	year         = 2018,
	booktitle    = {Phd Thesis}
}

@article{DKWB23,
	title        = {Subexponential-time algorithms for sparse PCA},
	author       = {Ding, Yunzi and Kunisky, Dmitriy and Wein, Alexander S and Bandeira, Afonso S},
	year         = 2023,
	journal      = {Foundations of Computational Mathematics},
	publisher    = {Springer},
	pages        = {1--50}
}

@inproceedings{BR13,
	title        = {Complexity theoretic lower bounds for sparse principal component detection},
	author       = {Berthet, Quentin and Rigollet, Philippe},
	year         = 2013,
	booktitle    = {Conference on learning theory},
	pages        = {1046--1066},
	organization = {PMLR}
}

@article{alon1998finding,
	title        = {Finding a large hidden clique in a random graph},
	author       = {Alon, Noga and Krivelevich, Michael and Sudakov, Benny},
	year         = 1998,
	journal      = {Random Structures and Algorithms},
	volume       = 13,
	number       = {3-4},
	pages        = {457--466}
}

@article{ZX18,
	title        = {Tensor SVD: Statistical and computational limits},
	author       = {Zhang, Anru and Xia, Dong},
	year         = 2018,
	journal      = {IEEE Transactions on Information Theory},
	publisher    = {IEEE},
	volume       = 64,
	number       = 11,
	pages        = {7311--7338}
}

@article{WBP16,
	title        = {Average-case hardness of RIP certification},
	author       = {Wang, Tengyao and Berthet, Quentin and Plan, Yaniv},
	year         = 2016,
	journal      = {Advances in Neural Information Processing Systems},
	volume       = 29
}

@article{WBS16,
	title        = {STATISTICAL AND COMPUTATIONAL TRADE-OFFS IN ESTIMATION OF SPARSE PRINCIPAL COMPONENTS},
	author       = {Tengyao Wang and Quentin Berthet and Richard J. Samworth},
	year         = 2016,
	journal      = {The Annals of Statistics},
	publisher    = {Institute of Mathematical Statistics},
	volume       = 44,
	number       = 5,
	pages        = {1896--1930},
	urldate      = {2023-07-31}
}

@article{MW15,
	title        = {COMPUTATIONAL BARRIERS IN MINIMAX SUBMATRIX DETECTION},
	author       = {Ma, Zongming and Wu, Yihong},
	year         = 2015,
	journal      = {The Annals of Statistics},
	publisher    = {JSTOR},
	pages        = {1089--1116}
}

@inproceedings{HWX15,
	title        = {Computational lower bounds for community detection on random graphs},
	author       = {Hajek, Bruce and Wu, Yihong and Xu, Jiaming},
	year         = 2015,
	booktitle    = {Conference on Learning Theory},
	pages        = {899--928},
	organization = {PMLR}
}

@article{GMZ17,
	title        = {SPARSE CCA: ADAPTIVE ESTIMATION AND COMPUTATIONAL BARRIERS},
	author       = {Chao Gao and Zongming Ma and Harrison H. Zhou},
	year         = 2017,
	journal      = {The Annals of Statistics},
	publisher    = {Institute of Mathematical Statistics},
	volume       = 45,
	number       = 5,
	pages        = {2074--2101},
	urldate      = {2023-07-31}
}

@article{Che15,
	title        = {Incoherence-optimal matrix completion},
	author       = {Chen, Yudong},
	year         = 2015,
	journal      = {IEEE Transactions on Information Theory},
	publisher    = {IEEE},
	volume       = 61,
	number       = 5,
	pages        = {2909--2923}
}

@misc{CW18,
	title        = {Statistical and Computational Limits for Sparse Matrix Detection},
	author       = {T. Tony Cai and Yihong Wu},
	year         = 2018,
	eprint       = {1801.00518},
	archiveprefix = {arXiv},
	primaryclass = {math.ST}
}

@article{CLR17,
	title        = {Computational and statistical boundaries for submatrix localization in a large noisy matrix},
	author       = {Cai, T Tony and Liang, Tengyuan and Rakhlin, Alexander and others},
	year         = 2017,
	journal      = {The Annals of Statistics},
	volume       = 45,
	number       = 4,
	pages        = {1403--1430}
}

@article{ding2025low,
  title={Low degree conjecture implies sharp computational thresholds in stochastic block model},
  author={Ding, Jingqiu and Hua, Yiding and Slot, Lucas and Steurer, David},
  journal={arXiv preprint arXiv:2502.15024},
  year={2025}
}

@article{brennan2020statistical,
  title={Statistical query algorithms and low-degree tests are almost equivalent},
  author={Brennan, Matthew and Bresler, Guy and Hopkins, Samuel B and Li, Jerry and Schramm, Tselil},
  journal={arXiv preprint arXiv:2009.06107},
  year={2020}
}
\appendix

\section{Proof of \cref{claim:approx-high-degree-part}}\label{sec:a2}

In this section we present the proof of \cref{claim:approx-high-degree-part}. Recall that our goal is to approximate $Tf_{-}(x)$, where $f(x) = \E_{A}[A(x)]$ is the average output of $A$. We will do this by computing $A$ on a polynomial number of inputs.

Recall that the function $f_+ = f-f_-$ is in the $q$-tractable (\cref{def:tractable-subspace}) subspace $V$, i.e., it can be written as $f_{+}=\sum_{0\leq i\leq \ell}c_i\cdot f_i$ for real-valued coefficients $c_i = \E_{\nul}[f_i\cdot f]$, where $\{f_1,\ldots, f_\ell\}$ is an orthonormal basis of $V$ that can be efficiently evaluated.

We begin by producing a large collection of independent samples $\mathbf{z}=\{z_1,\ldots, z_k\}$ from $\nul$. Define the function $g^\mathbf{z}(y) = A(y)-\sum_{0\leq i\leq \ell}f_i(y)\cdot \E_{j\sim [k]}[f_i(z_j)\cdot A(z_j)]$. We will argue that $g^z(y)$ is close to $f_{-}$. In the interest of conciseness, we will use very crude bounds.
\begin{align*}
    \E_{\mathbf{z}}[|g^{\mathbf{z}}(y) - f_{-}(y)|]&= \E_{\mathbf{z}}\left[\left|\sum_{0\leq i\leq \ell}f_i(y)\cdot (\E_{j\sim [k]}[f_i(z_j)\cdot A(z_j)] - \E_\nul[f_j\cdot f])\right|\right]\\
    &\leq \sum_{0\leq i\leq \ell}|f_i(y)|\cdot \E_\mathbf{z}\left[\left|\E_{j\sim [k]}[f_i(z_j)\cdot A(z_j)]-\E_\nul[f_i\cdot f]\right|\right]\\
    &\leq \sum_{0\leq i\leq \ell}|f_i(y)|\cdot \E_\mathbf{z}\left[\left(\E_{j\sim [k]}[f_i(z_j)\cdot A(z_j)]-\E_\nul[f_i\cdot f]\right)^2\right]^{1/2}\tag{Jensen's inequality}\\
    &= \sum_{0\leq i\leq \ell}|f_i(y)|\cdot \Var_\mathbf{z}\left[\E_{j\sim [k]}[f_i(z_j)\cdot A(z_j)]\right]^{1/2}\\
    &\leq k^{-1/2}\sum_{0\leq i\leq \ell}|f_i(y)|\cdot \Var_{z_1\sim \nul}\left[f_i(z_1)\cdot A(z_1)]\right]^{1/2}\\
    &\leq k^{-1/2}\cdot (\ell + 1)\cdot \max_{i}\|f_i\|_\infty^2\leq 2k^{-1/2}q^3.\tag{$\ell\leq q$, $\|f_i\|_\infty\leq q\cdot \|f_i\|_{2,\nul}= q$ by $q$-tractability}
\end{align*}

As long as $k\gg q^6/\eps^2$, this is at most $\eps/2$. Finally, we will approximate $Tf_-(x)$ by computing the empirical average of $g^\mathbf{z}(y)$ across a large number of independent samples $y\sim M(x)$. Concretely, sample another collection $\mathbf{y}=\{y_1,\ldots, y_k\}$ where $y_i\sim M(x)$ define $C(x) = C^{\mathbf{z},\mathbf{y}}(x) = \E_{j\sim [k]}[g^\mathbf{z}(y_j)]$.
For any $x\in \Omega$ and $\mathbf{z}$ we will bound
\begin{align*}\E_{\mathbf{y}}[|C(x)-Tg^\mathbf{z}(x)|] &\leq \E_{\mathbf{y}}[(C(x)-Tg^\mathbf{z}(x))^2]^{1/2} \\&= \Var_\mathbf{y}\left[\E_{j\sim [k]}[g^\mathbf{z}(y_j)]\right]^{1/2}\\
&=k^{-1/2}\Var_{y_1\sim M(x)}\left[\E_{j\sim [k]}[g^\mathbf{z}(y_1)]\right]^{1/2}\leq k^{-1/2}\|g^\mathbf{z}\|_\infty.\end{align*}
One can bound $\|g^\mathbf{z}\|_\infty\leq 2q^3$ for any $z$, so again this is at most $\eps/2$ for $k\gg q^6/\eps^2$. By the triangle inequality, we can conclude
\begin{align*}
  \E_{C}[|C(x)-Tf_{<\eps}(x)|]&\leq \E_{\mathbf{y},\mathbf{z}}[|C(x)-Tg^\mathbf{z}(x)|]+\E_{\mathbf{z}}[|Tf_{-}(x)-Tg^\mathbf{z}(x)|]\\
  &\leq \E_{\mathbf{y},\mathbf{z}}[|C(x)-Tg^\mathbf{z}(x)|]+\E_{y\sim M(x),\mathbf{z}}[|f_{-}(y)-g^\mathbf{z}(y)|]\\
  &\leq 2\cdot \eps/2 = \eps.
\end{align*}

\section{Details of Stochastic Gradient Descent Convergence}\label{sec:a1}

In this section we present the proof of \cref{claim:sgd}, which completes the proof of \cref{thm:general-2}. Let us begin by recalling the setup. We are trying to solve the following problem:
\[\min_{g=\sum_{i\in [\ell]}g_i f_i}\E_{x\sim \pla}[\sigma(1+g(x))].\]
Recall that the objective function is a convex function of the coefficients $g_1,\ldots, g_\ell$ of $g$. 
Throughout this section, we will use the notation $\|v\|_2 = \sqrt{\sum_{i\in [n]}v_i^2}$ exclusively for $n$-dimensional vectors. The goal is to find a point $g^* = \sum_i g_i^* f_i$ such that the $2$-norm of the gradient is bounded as follows.
\begin{equation}\label{eqn:grad-bound}
  \|\nabla_{\{g_1,\ldots, g_n\}}\E_{x\sim \pla}[\sigma(1+g(x))]\|_2\leq \delta/3.  
\end{equation}

We will use stochastic gradient descent to find an approximate optimizer. The bulk of this proof will be spent in verifying that the conditions for convergence of stochastic gradient descent are met. We will apply the following result.

\begin{theorem}[e.g. {\cite[Corollary~5.6]{garrigos2023handbook}}]\label{thm:sgd}
    Let $\mu$ be a distribution over differentiable convex functions $F:\R^n\to \R$ such that each function $F$ in the support of $\mu$ is $R$-smooth, i.e. for all $x,y\in \R^n$,
    \[\|\nabla F(x) - \nabla F(y)\|_2\leq R\cdot \|x - y\|_2.\]
    Say $\|\nabla F(x)\|_2\leq K$ for all $f$ in the support of $\mu$ and for all $x\in \R^n$.
    Suppose $\E_\mu[F]$ is minimized at $\hat{x}$. Assume that one can sample $F\sim\mu$, evaluate $F(x)$ and evaluate $\nabla F(x)$ in time $T$. Then there is a stochastic gradient descent based algorithm running in time $\poly(n, K, R, 1/\delta', T)$ that finds a point $x$ satisfying the following with high probability:
    \[\E_{F\sim \mu}[F(x)] \leq \E_{F\sim \mu}[F(\hat{x})] + \delta'.\]
\end{theorem}

Let us describe what parameters we will instantiate \cref{thm:sgd} with and verify that its conditions are met. We set $n=\ell$, so our variables are the coefficients $g_1,\ldots, g_\ell$ of $g = \sum_{i\in [\ell]}g_i f_i$. We will set $\mu$ to be the following distribution over functions: sample a point $x\sim \pla$, and output the function $F$ that maps $(g_1,\ldots, g_\ell)$ to $\sigma(1+g(x))$.

Note that we can write $\nabla F(g) = \sigma'(1+g(x))\cdot \mathbf{f}(x)$, where $\mathbf{f}(x)$ is the $\ell$-dimensional vector $[f_1(x),\ldots, f_\ell(x)]$. We can indeed sample from $\mu$ and evaluate $F$ and $\nabla F$ in time $T:=\poly(q)$. Using the fact that $\sigma'(\cdot)\in [0,1]$ and $\|f_i\|_\infty \leq q$, we will bound
\[\|\nabla F(g)\|_2^2=  \sigma'(1+g(x))^2\cdot \|\mathbf{f}(x)\|^2 \leq \sum_{i\in [\ell]} f_i(x)^2\leq \ell q^2\leq q^3.\]
So $K$ can be set to $q^3$. Similarly, we can bound
\begin{align*}
    \|\nabla F(g) - \nabla F(h)\|_2 &= \|(\sigma'(1+g(x))- \sigma'(1+h(x)))\cdot \mathbf{f}(x)\|_2\\
    &\leq \ell q^3\cdot |\sigma'(1+h(x)) - \sigma'(1+g(x))|\\
    &\leq O(\delta^{-1}q^3)\cdot |h(x) - g(x)|\tag{Lipschitz constant of $\sigma'$ is $\max_{t}|\sigma''(t)|=O(\delta^{-1})$}\\
    &\leq O(\delta^{-1}q^3)\cdot\sum_{i\in \ell}f_i(x)\cdot |h_i - g_i|\\
    &\leq O(\delta^{-1} q^3)\cdot \|\mathbf{f}(x)\|_2\cdot  \|h - g\|_2.\tag{Cauchy-Schwarz}\\
    &\leq O(\delta^{-1} q^{4.5})\cdot \|h - g\|_2.
\end{align*}
As a result, we can set $R = O(\delta^{-1}q^{4.5})$. Applying \cref{thm:sgd}, we get an algorithm running in time $\poly(q,1/\delta,1/\delta')$ that finds a point $g^*$ such that with high probability, $\E_{\nul}[\sigma(1+g^*)] \leq \E_{\nul}[\sigma(1+\hat{g})] + \delta'$.

Recall that we wanted an upper bound on the gradient (\cref{eqn:grad-bound}). We will use another off-the-shelf result to bound this.


\begin{lemma}[e.g. {\cite[Lemma~2.28]{garrigos2023handbook}}]
    If $F:\R^d\to \R$ is convex and $R$-smooth, then for all $x\in \R^d$, $\frac{1}{2R}\|\nabla F(x)\|_2^2\leq F(x) - F(\hat{x})$.
\end{lemma}

We immediately get $\|\nabla F(g^*)\|_2\leq O(\delta^{-1}q^{4.5})\cdot \delta'$. Choosing $\delta'$ to be small enough, this is at most $\delta/4$. To complete the proof, note that rounding the coefficients of $g^*$ to $O(\log (q/\delta))$ bits of precision can change $\|\nabla F(g^*)\|_2$ by at most $\poly(\delta/q)$. Therefore we can evaluate $g^*$ in time $\poly(q, \log 1/\delta)$ while still guaranteeing that $\|\nabla F(g^*)\|_2\leq \delta/2$, completing the proof.

\section{Calculation of Low-Degree advantage for Planted Clique}\label{sec:a3}
Let $\nul=G(n,1/2)$ and $\pla=G(n,1/2,k)$ for $k=o(\sqrt{n})$. The goal of this section is to calculate $\adv^{(\pla,\nul)}[\F_{\leq d}]^2$ up to $1+o(1)$ factors. We begin with the well-known expression $\adv^{(\pla,\nul)}[\F_{\leq d}]^2=\sum_{|S|\leq d}(k/n)^{2\cdot |V(S)|}$ (e.g. \cite{Hopkins18thesis}, Lemma 2.4.1). Let us bound
    \begin{align*}
        \adv^{(\pla,\nul)}[\F_{\leq d}]^2&=\sum_{|S|\leq d}(k/n)^{2\cdot |V(S)|}\\
        &= \binom{n}{2}\cdot \frac{k^4}{n^4} + \sum_{1<|S|\leq d}(k/n)^{2\cdot |V(S)|}\\
        &\leq \frac{k^4}{2n^2}+\sum_{\mathcal{S}}n^{|V(\mathcal{S})|}(k/n)^{2\cdot |V(\mathcal{S})|},
    \end{align*}
    where $\mathcal{S}$ ranges over all $\leq d$-vertex graphs on at least $3$ vertices. Since the number of $d$-vertex graphs is bounded by $2^{\binom{d}{2}}=O(1)$ and $k=o(\sqrt{n})$, we have
    \[\frac{k^4}{2n^2}\leq \adv^{(\pla,\nul)}[\F_{\leq d}]^2 \leq  \frac{k^4}{2n^2}+ O\left(\frac{k^6}{n^3}\right)=(1+o(1))\cdot \frac{k^4}{2n^2}.\]

\end{document}